\documentclass[UKenglish]{lmcs}

\usepackage[T1]{fontenc}
\usepackage{amsmath,amssymb,stmaryrd}
\usepackage{tikz-cd}
\usepackage{stmaryrd}[only,leftrightarroweq,llbracket,rrbracket]
\usepackage[linkcolor={blue},citecolor={blue},urlcolor={blue},breaklinks=true,final,bookmarks,unicode,colorlinks=true]{hyperref}
   \def\@citecolor{blue}%
   \def\@urlcolor{blue}%
   \def\@linkcolor{blue}%

\usepackage[noabbrev,capitalise,nameinlink]{cleveref}
\usepackage{enumitem}
\usepackage{proof}
\usepackage{stackrel}
\usepackage{mathtools}
\usepackage{orcidlink}

\usepackage[layout=footnote,marginclue,final]{fixme}
\FXRegisterAuthor{bk}{abk}{BK}	% Barbara
\FXRegisterAuthor{km}{akm}{KM}	% Karla
\FXRegisterAuthor{ls}{als}{LS}	% Lutz
\FXRegisterAuthor{jf}{ajf}{JF}	% Jonas
\FXRegisterAuthor{pw}{apw}{PW}	% Paul

\usepackage[appendix=strip,bibliography=common]{apxproof}

\theoremstyle{definition}
\newtheorem{defn}[thm]{Definition}
\newtheorem{expl}[thm]{Example}

\Crefname{expl}{Example}{Examples}
\Crefname{lem}{Lemma}{Lemmas}
\Crefname{thm}{Theorem}{Theorems}
\Crefname{defn}{Definition}{Definitions}

\newcommand*{\myby}[1]{(\text{#1})}

\newcommand{\dfun}{\mathcal{D}}
\newcommand{\pfun}{\mathcal{P}}
\newcommand{\cfun}{\mathcal{C}}

\newcommand{\catC}{\mathbf{C}}
\newcommand{\set}{\mathbf{Set}}

\newcommand{\pmet}{\mathbf{PMet}}
\newcommand{\meas}{\mathbf{Meas}}
\newcommand{\met}{\mathbf{Met}}

\newcommand{\kantrel}[1]{K^\mathsf{rel}_{#1}}
\newcommand{\kantsym}[1]{K_{#1}}
\newcommand{\wass}[1]{W_{#1}}
\newcommand{\wasspmet}[1]{W^\mathsf{PMet}_{#1}}
\newcommand{\dk}{\delta^\mathsf{KR}} % Kantorovich-Rubinstein metric
\newcommand{\dkrel}{\delta^\mathsf{KR,rel}} % Kantorovich-Rubinstein metric
\renewcommand{\dh}{\delta^\mathsf{H}} % Hausdorff metric
\newcommand{\dhk}{\delta^\mathsf{HK}} % Hausdorff-Kantorovich metric
\newcommand{\dlp}{\delta^\mathsf{LP}} % Lévy-Prokhorov metric
\newcommand{\de}{d_e} % Euclidean distance 

\newcommand{\nat}{\mathbb{N}}
\newcommand{\real}{\mathbb{R}}
\newcommand{\preal}{[0,1]}
\newcommand{\Rat}{\mathbb{Q}}
\renewcommand{\epsilon}{\varepsilon}
\renewcommand{\phi}{\varphi}
\newcommand{\bfunctional}[1]{\Phi_{#1}} % Functional for which behavioiural distance is a fixpoint
\newcommand{\bmetric}[1]{\mu\bfunctional{#1}} % Functional for which behavioiural distance is a fixpoint
\newcommand{\To}{\Rightarrow}

\newcommand{\into}{\hookrightarrow}

\newcommand{\id}{\mathrm{id}}

\newcommand{\relto}{\mathbin{\ooalign{$\rightarrow$\cr$\hspace{0.12ex}+$\cr}}}
\newcommand{\rev}[1]{{#1}^\circ}

\newcommand*{\expect}[2]{\mathbb{E}_{#1}({#2})}
\newcommand*{\expectSymb}{\mathbb{E}}
\newcommand*{\ball}[3]{B^{#1}_{#2}({#3})}
\newcommand*{\convx}{\mathsf{conv}}

\newcommand*{\Lip}{\mathsf{Lip}} % Lipschitz functions
\newcommand*{\Rad}{\mathsf{Rad}} % Radon measures
\newcommand*{\lipnorm}[1]{\|#1\|_{\mathsf{Lip}}}
\newcommand*{\opnorm}[1]{\|#1\|_{\mathsf{op}}}
\newcommand*{\sint}[3]{{\textstyle\int_{#1} {#2} \,\dif {#3}}}
\newcommand*{\dif}{\mathrm{d}}

\newcommand{\sem}[1]{\llbracket #1\rrbracket}
\newcommand{\FLang}{\mathcal{F}}

\title[Case Studies in Generalized Kantorovich-Rubinstein Duality]{Case Studies in\\Generalized Kantorovich-Rubinstein Duality}

\thanks{The authors acknowledge support by the Deutsche Forschungsgemeinschaft (DFG, German Research Foundation).
  The second author has been supported by project number 531706730 (CoRSA), while the remaining authors have been supported by project number 434050016 (SpeQt).}

\author[P.~Wild]{Paul Wild\lmcsorcid{0000-0001-9796-9675}}[a]
\author[L.~Schröder]{Lutz Schröder\lmcsorcid{0000-0002-3146-5906}}[a]
\author[P.~Nora]{Pedro Nora\lmcsorcid{0000-0001-8581-0675}}[b]
\author[B.~König]{Barbara König\lmcsorcid{0000-0002-4193-2889}}[b]
\author[J.~Forster]{Jonas Forster\lmcsorcid{0000-0002-5050-2565}}[a]

\address{Friedrich-Alexander-Universität Erlangen-Nürnberg, Erlangen, Germany}
\email{paul.wild@fau.de, lutz.schroeder@fau.de, jonas.forster@fau.de}

\address{Universität Duisburg-Essen, Duisburg, Germany}
\email{pedro.nora@uni-due.de, barbara\_koenig@uni-due.de}

\begin{document}
       
\begin{abstract}
  The classical Kantorovich-Rubinstein duality guarantees coincidence
  between metrics on the space of probability distributions defined on
  the one hand via transport plans (\emph{couplings}) and on the other
  hand via price functions. Both constructions have been lifted to the
  level of generality of set functors, with the construction based on
  couplings referred to as the \emph{Wasserstein} or simply the
  \emph{coupling-based} lifting, and the price-function-based
  construction as the \emph{Kantorovich} or \emph{codensity} lifting,
  both based on a choice of quantitative modalities for the given
  functor. It is known that every coupling-based lifting can be
  expressed as a price-function-based lifting; however, the latter in
  general needs to
  use additional modalities. We give an example showing that this
  cannot be avoided in general. We refer to cases in which the same
  modalities can be used as satisfying the \emph{generalized
  Kantorovich-Rubinstein duality}. We establish the generalized
  Kantorovich-Rubinstein duality in this sense for two important
  cases: The L\'evy-Prokhorov distance on distributions, which finds
  wide-spread applications in machine learning due to its favourable
  stability properties, and the standard metric on convex sets of
  distributions that arises by combining the Hausdorff and
  Kantorovich-Rubinstein distances.

  Furthermore we generalize (coupling-based) functor liftings from set
  functors to functors which act on (pseudo-)metric spaces, thus
  incorporating new example functors that map to closed subsets
  respectively Radon measures.
\end{abstract}

\maketitle

\section{Introduction}
\label{sec:introduction}

\pwnote[inline]{Find a new title. New angle might be that we clarify what generalized Kantorovich-Rubinstein amounts to in three different settings: for pseudometric liftings, for fuzzy-relational liftings, and for liftings of functors $\pmet\to\set$ or $\meas\to\set$ (like the Borel versions of Lévy-Prokhorov and convex powerset)}

\pwnote{Move the treatment of functors $\pmet/\meas\to\set$ to its own section. Cite duality results for Kantorovich over Borel and Hausdorff over closed/compact sets, then do the proofs for Lévy-Prokhorov and convex powerset.}
\pwnote{Prove that $K \le W$ holds in the pseudometric setting as well}
\pwnote{Intentionally left as open: measure-theoretic-versions of fuzzy-relational duality}

Measuring behavioural distances between probabilistic systems requires
notions of distance between probability distributions
(e.g.~\cite{bw:behavioural-pseudometric}). One well-established metric
on the set of distributions over a metric space is variously termed
the \emph{Kantorovich-Rubinstein}, \emph{Wasserstein}, or \emph{Hutchinson}
metric. It can be calculated either by minimizing over the expected
value of \emph{transport plans} between, or \emph{couplings} of, the
given distributions, or by maximizing over the difference of
expectations taken over all nonexpansive \emph{price functions}. The
coincidence of these two values is the classical
\emph{Kantorovich-Rubinstein
  duality}~\cite[Theorem~5.10]{v:optimal-transport}. Intuitively
speaking, a transport plan or coupling is a way to transform one
distribution into another by shifting around weight, and its cost, to
be minimized, is determined by how much weight is shifted over which
distances. On the other hand, a price function determines a price for
some commodity at given points; nonexpansiveness of the price
function means that no profit can be made from full-cost transport.
The difference between the expected values of a price function under
the given distributions is the profit to be made by having the
commodity transported, and hence the amount one can offer to a
logistics provider when outsourcing the transport.

Both the coupling-based definition and the price-function-based
definition have been generalized categorically to construct liftings
of set functors to the category of \mbox{(pseudo-)}metric
spaces~\cite{bbkk:coalgebraic-behavioral-metrics} and quantitative lax
extensions of set functors~\cite{FuzzyLaxHemi}, where the latter are
distinguished by applying to unrestricted quantitative relations
instead of only to \mbox{(pseudo-)}metrics. Metric functor liftings and lax
extensions in particular serve to give a general treatment of
\emph{behavioural distances} on quantitative systems such as
probabilistic, weighted, or metric~\cite{afs:linear-branching-metrics}
transition systems in the framework of \emph{universal
coalgebra}~\cite{Rutten00}. In this framework, functors serve as
parameters determining a type of systems as their coalgebras; for
instance, coalgebras for the distribution functor are Markov chains.
In the generalized setting, the coupling-based construction is often
referred to as the \emph{Wasserstein} lifting or extension, and the
price-function-based one as the \emph{Kantorovich} lifting or
extension, but other names are found in the literature, such as
\emph{coupling-based}~\cite{HumeauEA25} and
\emph{codensity}~\cite{KomoridaEA19} lifting, respectively.
Both constructions are parametrized over a choice of
quantitative modalities; the classical case involves, on both sides,
only one modality, the expectation modality.

The interest in having a price-function-based presentation of a given functor
lifting or lax extension lies inter alia in the fact that one obtains
a quantitative Hennessy-Milner property for the quantitative modal
logic generated by the respective
modalities~\cite{km:bisim-games-logics-metric,FuzzyLaxHemi}. This
property states coincidence of the behavioural distance induced by the
given lifting or extension with the logical distance induced by the
respective quantitative modal logic; thus, high distance between two
states can always be \emph{certified} by means of a modal formula, a
prominent principle in the study of behavioural distances (e.g.~\cite{bw:behavioural-pseudometric,DesharnaisEA08,FuzzyLaxHemi,km:bisim-games-logics-metric,Forster_et_al:CSL.2023:Density}, and more recently~\cite{rb:explainability-labelled-mc,TurkenburgEA26,ForsterEA26}).  It has been shown that
every metric functor lifting that preserves
isometries~\cite{KantorovichFunctors} and every quantitative lax
extension~\cite{FuzzyLaxHemi} \emph{is Kantorovich}, i.e.~can be
presented via the generalized price-function-based construction, using however
a rather large (in particular typically infinite) set of modalities
called the \emph{Moss modalities}.

An important open point that remains is thus the question of what we
term \emph{generalized Kantorovich-Rubinstein duality}: In which cases
does a coupling-based distance given by a choice of modalities for a
functor coincide with the price-function-based distance \emph{for the same
  modalities}? Known positive examples include, as mentioned, the
classical Kantorovich-Rubinstein distance on distributions, but also the
Hausdorff distance on the powerset. We begin our analysis by giving
an example of a natural coupling-based metric for which
generalized Kantorovich-Rubinstein duality in this sense fails, namely
the standard \emph{$p$-Wasserstein} metric for $p>1$ (which minimizes
$p$-th roots of the expectation of $p$-th powers of couplings). As our
main contribution, we then provide two new positive examples, namely
the \emph{L\'evy-Prokhorov} distance on probability
distributions~\cite{Prokhorov56} and the standard distance on the
\emph{convex powerset}, whose elements are convex sets of probability
distributions. The L\'evy-Prokhorov distance has seen a recent rise in
popularity due to its favourable robustness properties that make it
suitable for tasks in machine learning (such as conformal
prediction~\cite{AolariteiEA25} and corruption
resistance~\cite{BennounaEA23}); moreover, it has been shown in recent
work~\cite{DesharnaisSokolova26} to induce precisely the behavioural distance
defined by
$\epsilon$-bisimilarity~\cite{DesharnaisEA08}. In this case, the
relevant modality is precisely the \emph{generally} modality used in
work on fuzzy description logics as an alternative formal
correspondent of the natural-language term
`probably'~\cite{SchroderPattinson11}. The convex powerset plays a
central role in the distribution semantics of Markov decision
processes (or probabilistic
automata)~\cite{Bonchi_et_al:PowerConvexAlg}. Its standard metric is
just the composition of the Hausdorff and Kantorovich-Rubinstein metrics, and as
such is again given via the coupling-based construction for a modality
composed from the expectation modality and the standard fuzzy diamond
modality~\cite{FuzzyLaxHemi}. We show that the generalized
Kantorovich-Rubinstein duality holds w.r.t.~this modality. Beyond the
mentioned benefits regarding characteristic quantitative modal logics,
we demonstrate in this case that the new price-function-based description also
allows for more efficient computation of distances.

This work is an extended and revised version of a previous conference
publication~\cite{WildEA26}. Beyond including full proofs and
additional discussion of context, in particular of characteristic
quantitative modal logics, the material presented extends the scope of
the previous results: Regarding the case studies, we generalize the
duality result for L\'evy-Prokhorov distance from discrete
distributions to Radon distributions. In terms of the general
framework, we take the opportunity to formulate the notion of
generalized Kantorovich-Rubinstein duality abstractly for more general
functors than set functors, specifically for functors from the
category of pseudometric spaces to the category of sets. To this end,
we in particular introduce a coupling-based construction of functor
liftings and lax extensions at this higher level of generality.

\subsubsection*{Related Work} 
Categorical distance constructions based on couplings were first considered in work on monoidal topology~\cite{Hofmann07,HofmannEA14}.
The categorical price-function-based construction is predated by several constructions for specific functors in work on stochastic games~\cite{GoubaultLarrecq08}.
Our price-function-based presentation of convex
powerset complements earlier results on a coupling-based
characterization~\cite{FuzzyLaxHemi} and a quantitative-algebraic
presentation~\cite{mv:monads-quantitative-equational}.
The modality underlying the distance on the convex powerset functor relates to Goubault-Larrecq's \emph{previsions}~\cite{GoubaultLarrecq07Previsions};  distances between such previsions are  also expressed through a combination of (topological versions of) the Hausdorff and Kantorovich-Rubinstein-distances~\cite{GoubaultLarrecq08,GoubaultLarrecq17,GoubaultLarrecqKRQM-I}.
There has been
recent work on what are termed \emph{correspondences} between the
generalized coupling-based and price-function-based constructions where a single
modality is assumed for the former, while an
associated set of modalities is considered for the latter~\cite{HumeauEA25}. Correspondences in this sense thus lie
between dualities as considered here, where we insist on the same
modalities being used on both sides, and general theorems on
Kantorovich presentations of lax extensions~\cite{FuzzyLaxHemi} and
classes of functors~\cite{KantorovichFunctors} that use very large
sets of liftings. Results are obtained for functors constructed from
the main known instances (distributions with the standard
Kantorovich-Rubinstein
metric, powerset) by applying coproduct and product. The price-function-based
construction is sometimes referred to as the \emph{codensity}
construction~\cite{KomoridaEA19,kkkrh:expressivity-quantitative-modal-logics,HumeauEA25},
and as such has been used for logical characterizations of behavioural
distances as mentioned
above~\cite{kkkrh:expressivity-quantitative-modal-logics} but also for
game characterizations~\cite{KomoridaEA19}. The problem of generalized
Kantorovich-Rubinstein duality has been stated already in work
introducing coalgebraic behavioural
distances~\cite{bbkk:coalgebraic-behavioral-metrics}, and a simple
counterexample has been given; the counterexample we give here is
distinguished by involving a quantitative modality that satisfies an
analogue of two-valued \emph{separation}~\cite{Pattinson04}.

\pwnote{Also cite non-recent work on behavioural distances}
% \begin{itemize}
%   \item We care about this both over pseudometrics and over fuzzy relations.
%     There are some differences here, e.g.\ in the Hausdorff case we have $\kant{\inf} = \wass{\inf}$ over pseudometrics and $\kant{\{\sup,\inf\}} = \wass{\inf}$ over fuzzy relations.
%   \item In all known cases where $\Lambda = \{\lambda\}$ we have that $\lambda$ is one of the Moss liftings. (Remark on this somewhere)
% \end{itemize}

\section{Preliminaries}
\label{sec:preliminaries}
\noindent We discuss preliminaries on (pseudo-)metric spaces, fuzzy relations
and coalgebras. Generally, we assume basic familiarity with category
theory~\cite{AdamekEA90}.

\subsubsection*{Metric spaces} We write
$\oplus, \ominus\colon [0,1] \times [0,1] \to [0,1]$ for truncated
addition and subtraction, i.e. $a \oplus b = \min\{1, a+b\}$ and
$a \ominus b = \max\{0, a-b\}$.  A \emph{($1$-bounded) pseudometric
  space} is a pair $(X, d_X)$, where $X$ is a set and
$d_X \colon X \times X \to \preal$ is a function, which for all
$x,y,z \in X$ is subject to the conditions of \emph{reflexivity}
$d_X(x, x) = 0$, \emph{triangle inequality}
$d_X(x,z)\leq d_X(x, y) + d_X(y, z)$ and \emph{symmetry}, that is
$d_X(x,y) = d_X(y,x)$.  A \emph{metric space} is then a pseudometric
space which is separated: If $d_X(x,y) = 0$ then $x = y$. The
\emph{Euclidean distance} $d_e(a, b) = |b-a|$ makes $\preal$ into a
metric space.

A function between the underlying sets $f\colon X \to Y$ of two
pseudometric spaces $(X, d_X)$, $(Y, d_Y)$ is \emph{nonexpansive} if
distances are not increased by $f$, explicitly if for all $x, y\in X$
it holds that $d_Y(f(x), f(y)) \leq d_X(x, y)$. Pseudometric spaces
and nonexpansive functions between them form a category $\pmet$. The
full subcategory of $\pmet$ spanned by metric spaces is denoted by
$\met$.

The \emph{(categorical) product} $(X,d_X)\times(Y,d_Y)$ of two pseudometric spaces $(X,d_X)$ and $(Y,d_Y)$ is the pseudometric space with underlying set $X\times Y$ and pseudometric $d_\times((x_1,y_1),(x_2,y_2)) = \max(d_X(x_1,x_2),d_Y(y_1,y_2))$. This pseudometric is sometimes also called \emph{supremum distance}.
The \emph{Manhattan product} $(X,d_X)\boxplus(Y,d_Y)$ of two pseudometric spaces $(X,d_X)$ and $(Y,d_Y)$ is another pseudometric space on the set $X\times Y$, where the pseudometric is given by $d_\boxplus((x_1,y_1),(x_2,y_2)) = d_X(x_1,x_2) \oplus d_Y(y_1,y_2)$.

\subsubsection*{Fuzzy relations}
A \emph{fuzzy relation} $r\colon X\relto Y$ is a map $X\times Y\to[0,1]$.
Given fuzzy relations $r\colon X\relto Y$ and $s\colon Y\relto Z$, their \emph{composite} $s\cdot r$ is given by $(s\cdot r)(x,z) = \inf_{y\in Y} r(x,y) \oplus s(y,z)$.
Every function $f\colon X\to Y$ can be seen as a fuzzy relation (denoted by the same letter if the type is clear from the context) $f\colon X\relto Y$, where $f(x,y) = 0$ for the fuzzy relation iff $f(x) = y$ for the function, and $f(x,y) = 1$ otherwise.
In this way, identity functions $\id_X\colon X\relto X$ also act as identity with respect to relational composition.
The \emph{converse} of a fuzzy relation $r\colon X\relto Y$ is the fuzzy relation $\rev{r}\colon Y\relto X$, given by $\rev{r}(y,x) = r(x,y)$.
Fuzzy relations are ordered pointwise, that is $r\le r'$ iff $r(x,y) \le r'(x,y)$ for all $x$, $y$.

\subsubsection*{Coalgebra} Our main results derive some of their
interest from their relevance to behavioural distances in
coalgebras. Generally, the framework of \emph{universal
  coalgebra}~\cite{Rutten00} is based on abstracting state-based
systems as \emph{$F$-coalgebras} for a functor
$F\colon \catC \to \catC$ on a category~$\catC$, with~$F$ determining
the \emph{type} of the system. Specifically, an $F$-coalgebra is a
pair $(C, \gamma)$ consisting of a $\catC$-object $C$, thought of as
an object of \emph{states}, and a morphism $\gamma\colon C \to FC$
determining \emph{transitions} from states to structured collections
of successor states, with the structure determined by~$F$. A
\emph{homomorphism} between $F$-coalgebras $(C,\gamma)$ and
$(D,\delta)$ is a $\catC$-morphism $h \colon C \to D$ such that
$\delta \circ h = Fh \circ \gamma$.

We list some common functors that will be useful in the further technical development and their associated coalgebras.

\begin{expl}\label{exmpl:common-functors}
  \begin{enumerate}[wide]
  \item The (covariant) \emph{powerset functor}
    $\pfun\colon \set\to \set$ sends each set to its powerset. On
    functions, $\pfun$ acts by taking direct images: For
    $A \in \pfun X$ and $f\colon X \to Y$, we have
    $\pfun f (A) = f[A]$. Its coalgebras are precisely sets equipped
    with a binary relation, i.e.~\emph{transition systems} or
    \emph{Kripke frames}.
  \item The \emph{finitely supported probability distribution functor}
    $\dfun \colon \set \to \set$ sends a set $X$ to the set
    \[\dfun X = \{\mu \colon X \to [0,1] \mid \mu(x) > 0 \text{ for
        finitely many } x\in X \text{ and } \textstyle\sum_{x\in X}
      \mu(x) = 1\}.\] We can equivalently describe $\mu\in\dfun X$ by
    its action on subsets $A\subseteq X$, given as expected by
    $\mu(A)=\sum_{x\in A}\mu(x)$. On a function $f\colon X \to Y$, the
    functor $\dfun$ then acts by taking image measures in the standard
    sense:
    \[\dfun f(\mu)(A) = \textstyle\mu(f^{-1}[A]).\]
    The coalgebras of~$\dfun$ are precisely (discrete-time) Markov chains.
  \end{enumerate}
\end{expl}

\section{Functor Liftings and Lax Extensions}
\label{sec:lift-and-lax}

A central question in the study of state-based systems at large is
whether two states exhibit the same behaviour. In universal coalgebra,
answers for this type of question are provided by such concepts as
Aczel-Mendler bisimulation or behavioural equivalence. When the
behaviour of states has quantitative aspects, however, such as
probabilistic transitions or outputs in a metric space, small
deviations in these quantities immediately render two states
behaviourally distinct under such two-valued notions. When one prefers
to retain the information that these states differ only slightly, an
established approach, discussed next, is to switch from behavioural
equivalence relations to the more robust concept of \emph{behavioural
  metrics}, equipping the state space with a pseudometric structure to
describe how dissimilar individual states are in their behaviour. 

Behavioural metrics are typically defined either in terms of functor liftings or in terms of lax extensions.
In both cases, we work with a functor $F\colon\set\to\set$ encapsulating the branching type of the transition systems at hand.

In the functor lifting approach, we use a pseudometric on a set $X$ to define a pseudometric on the set $FX$.
This gives rise to a functor $\overline{F}\colon\pmet\to\pmet$, which is a lifting of $F$ in the sense that $U\overline{F} = FU$, where $U\colon\pmet\to\set$ is the forgetful functor.

Given a set $X$, we denote by $\pmet_X$ the set of pseudometrics $d\colon X\times X\to[0,1]$.
$\pmet_X$ is a complete lattice under the pointwise order.
If $X$ is equipped with a coalgebra structure $\gamma\colon X\to FX$, then we can use the lifted functor $\overline{F}$ to define a map $\bfunctional{\gamma}\colon \pmet_X\to \pmet_X$ on, given by
\begin{equation*}
  \bfunctional{\gamma}(d_X)(x, y) = d_{\overline F(X,d_X)}(\gamma(x), \gamma(y)).
\end{equation*}
As the map~$\bfunctional{\gamma}$ is monotone, it has a least fixpoint $\bmetric{\gamma}$ by the Knaster-Tarski fixpoint theorem, called the~\emph{behavioural distance}~\cite{bbkk:coalgebraic-behavioral-metrics} on $(X, \gamma)$.
This is the quantitative analogue of defining bisimilarity as a greatest fixpoint.
Note that we use a least fixpoint here instead because of the order reversal in the pseudometric setting: lower distances correspond to a greater degree of relatedness between states.\pwnote{Maybe move the behavioural distance stuff elsewhere. LS: OK like this, I would say}

In the lax extension approach, we generalize from pseudometrics $X\times X\to[0,1]$ to arbitrary fuzzy relations $r\colon X\times Y\to[0,1]$, and lift them to fuzzy relations $Lr\colon FX\times FY\to[0,1]$.
The construction~$L$ is a \emph{lax extension} if it satisfies the following conditions:
\begin{align*}
  \text{(L1)}\quad & r \le r' \implies Lr \le Lr' \\
  \text{(L2)}\quad & L(s\cdot r) \le Ls \cdot Lr \\
  \text{(L3)}\quad & Ff \le Lf \text{ and } \rev{(Ff)} \le L(\rev{f}) \\
  \intertext{If $L$ additionally satisfies the following axiom, then we say that it \emph{preserves converse}:}
  \text{(L0)}\quad & L(\rev{r}) = \rev{(Lr)} 
\end{align*}
It is easy to see that if the lax extension $L$ preserves converse, then $Ld_X\colon FX\times FX\to[0,1]$ is a pseudometric for every pseudometric $d_X\colon X\times X\to[0,1]$, so that $L$ gives rise to a functor lifting of $F$ to pseudometric spaces.
Without preservation of converse, $L$ lifts to hemimetric spaces instead~\cite{FuzzyLaxHemi}, which are like pseudometric spaces but without the symmetry axiom.

\begin{expl}
  \begin{enumerate}[wide]
    \item The \emph{Hausdorff lifting} $\overline \pfun \colon \pmet \to \pmet$ equips the
      powerset $\pfun X$ of the carrier of a metric space $(X, d_X)$ with
      the Hausdorff metric $\dh(d_X)$. The distance of two subsets $A, B \in \pfun X$
      is then given by
      \begin{equation}\label{eqn:dh}
        \dh(d_X)(A, B) := \max (
          \adjustlimits{\sup}_{x \in A}{\inf}_{y\in B} d_X(x, y),
          \adjustlimits{\sup}_{y \in B}{\inf}_{x\in A} d_X(x, y))  
      \end{equation}
  \item The \emph{Kantorovich-Rubinstein lifting} $\overline \dfun \colon \pmet
      \to \pmet$ of $\dfun$ equips the set $\dfun X$ with the
      \emph{Kantorovich-Rubinstein distance} $\dk(d_X)$, which is defined as
      \begin{equation}\label{eqn:dk}
        \dk(d_X)(\mu, \nu) := \sup \{
        \expect{\nu}{f} - \expect{\mu}{f}
        \mid f\colon (X,d_X)\to([0,1],\de) \text{ nonexpansive}\},
      \end{equation}
  where $\expect{\mu}{f} = \sum_{x\in X} \mu(x)\cdot f(x)$ denotes the expected value of $f$ under $\mu$.
  \end{enumerate}
\end{expl}

\noindent
Both of these examples can also be cast as lax extensions, resulting in the \emph{Hausdorff extension} and the \emph{Kantorovich-Rubinstein extension}, respectively~\cite{FuzzyLaxHemi}.
In the case of the former, the construction~\eqref{eqn:dh} obviously equally applies to fuzzy relations between different sets (and as such does indeed satisfy the laws of a lax extension), but in the case of the latter one has to find a suitable analogue of the notion of nonexpansive function:
\begin{defn}\label{defn:nonexp-pair}
  Let $r\colon X\times Y\to[0,1]$ be a fuzzy relation.
  We say that a pair $(f,g)$ of functions $f\colon X\to[0,1]$, $g\colon Y\to[0,1]$ is \emph{$r$-nonexpansive} if for all $x\in X$ and $y\in Y$,
  \begin{equation*}
    g(y) - f(x) \le r(x,y).
  \end{equation*}
\end{defn}
\noindent
The \emph{Kantorovich-Rubinstein extension} of $r\colon X\times Y\to[0,1]$ is then defined as
\begin{equation*}
  \dkrel(r)(\mu, \nu) := \sup \{
      \expect{\nu}{g} - \expect{\mu}{f}
      \mid (f,g) \text{ $r$-nonexpansive}\}.
\end{equation*}

\noindent
Many versions of the Hausdorff and Kantorovich-Rubinstein constructions found in the literature are not covered by the set-functor based constructions we have seen so far, as they use the underlying pseudometric of the space in a more far-reaching way.
Specifically, the Hausdorff distance is often restricted to act on closed or compact sets (to ensure that metrics are sent to metrics), while the Kantorovich-Rubinstein distance is defined on more general Borel or Radon measures instead of just the discrete ones.

To accommodate these constructions in the categorical setup, we replace $\set$ as the domain category by $\pmet$.
More explicitly, given a functor $F\colon\pmet\to\set$, we are interested in functor liftings $\overline{F}\colon\pmet\to\pmet$ such that $U\overline{F} = F$.
For the aforementioned examples we then consider the functors $\pfun^\mathsf{c},\dfun^\Rad\colon\pmet\to\set$, which respectively map a pseudometric space to the sets of its compact subsets and Radon measures, and then define generalized constructions lifting these to functors $\overline{\pfun^\mathsf{c}},\overline{\dfun^\Rad}\colon\pmet\to\pmet$.

In the powerset case, the corresponding Hausdorff construction $\dh$ is in fact defined just as before (\eqref{eqn:dh}), with the only change being that the domain is restricted to closed sets~$A$ and~$B$.
The Kantorovich-Rubinstein construction $\dk$ is also the same as before (\eqref{eqn:dk}), noting that every non-expansive map is continuous and hence Borel-measurable, meaning that the expected values are defined.

\subsubsection*{Transportation Theory}
All three versions of the Kantorovich-Rubinstein construction admit an alternative representation in terms of \emph{couplings}.
For simplicity, we describe this in the discrete case, where a coupling of probability distributions $\mu,\nu\in\dfun X$ is a probability distribution $\rho\in\dfun(X\times X)$ that has $\mu$ and $\nu$ as its \emph{marginals}, that is, $\dfun\pi_1(\rho)=\mu$ and $\dfun\pi_2(\rho)=\nu$, where $\pi_1,\pi_2\colon X\times X\to X$ are the left and right projections.
Then 
\begin{equation*}
  \dk(d_X)(\mu,\nu) = \inf \{ \expect{\rho}{d_X} \mid \rho \text{ is a coupling of $\mu$ and $\nu$} \}.
\end{equation*}
This representation is \emph{dual} in the sense that it characterizes the distance value as an infimum while the original definition~\eqref{eqn:dk} characterizes it as a supremum.
The fact that the two representations coincide is known as \emph{Kantorovich-Rubinstein duality}~\cite{Kantorovich39}.

The coupling-based presentation is closely related to optimal
transportation theory~\cite{v:optimal-transport}: 
One may view probability distributions $\mu$, $\nu$ as producers
and consumers of a resource respectively (with the amount of the
resource produced/consumed being fixed to $1$). Now one wants to
transport the produced resources to the consumers in the economically
most efficient possible way, minimizing the average distance each unit
of resource needs to travel. Couplings $\rho$ of $\mu$ and $\nu$ can then be
seen as \emph{transport plans}, with $\expect{\rho}{d_X}$ giving the
total cost of carrying out the plan. Then the Kantorovich-Rubinstein
distance gives us by definition the minimal possible cost.

% \subsubsection*{Expressive Logics and Distinguishing Formulae}
% \label{sec:distinguishing}

% \pwnote[inline]{Discuss how duality results amount to providing an expressive logic (via the $\Lambda$) for a predicate lifting $\lambda$.
% It might be better to discuss this only after the main examples.
% }
The original presentation~\eqref{eqn:dk} can be explained via the analogy of
``outsourcing the transport'' by defining a function~$f$ that assigns
a price to each $x\in X$. This function must satisfy the requirement
that it is nonexpansive, i.e., that the difference of prices
assigned to $x,y$ is always at most $d_X(x,y)$ (meaning that no extra
profit can be made from such a transport). The overall profit under
such a \emph{price function} is then the income obtained from the consumers
($\expect{\nu}{f}$) minus the cost paid to the producers
($\expect{\mu}{f}$). Taking the maximum over all such $f$ gives us the
value of the Kantorovich-Rubinstein lifting.

In the next section, we discuss general constructions of functor liftings and lax extensions based on the two sides of the Kantorovich-Rubinstein duality.

\section{Generalized Constructions based on Predicate Liftings}
\label{sec:constructions}

The two dual representations of the Kantorovich-Rubinstein lifting and extension give rise to more general constructions.
These constructions are based on the observation that the expectation operation is an example of a \emph{(real-valued) predicate lifting}: given a real-valued predicate on a set $X$ (that is, a map $X\to[0,1]$), taking expected values results in a real-valued predicate on the set $\dfun X$ of probability distributions on $X$ (a map $\dfun X\to[0,1])$.
More generally, a predicate lifting for a functor $F$ also operates on such predicates and turns them into predicates on $FX$.
We consider two kinds of predicate liftings for different kinds of functors:
\begin{defn}\label{defn:pred-lifting}
  \begin{enumerate}[wide]
    \item Let $F\colon\set\to\set$ be a functor.
      A \emph{$\preal$-valued predicate lifting} for~$F$ is a natural transformation $\lambda$ whose components are maps
      \begin{equation*}
        \lambda_X\colon \set(X,[0,1]) \to \set(FX,[0,1]).
      \end{equation*}
    \item Let $F\colon\pmet\to\set$ be a functor.
      A \emph{$\preal$-valued predicate lifting} for $F$ is a natural transformation $\lambda$ whose components are maps
      \begin{equation*}
        \lambda_{(X,d)}\colon \pmet((X,d),([0,1],d_e)) \to \set(FX,[0,1]).
      \end{equation*}
  \end{enumerate}
\end{defn}

\noindent
In the case of functors defined on $\pmet$, therefore, every predicate to be lifted is required to be non-expansive. (Another variant that has appeared in the literature are predicate liftings for functors $\pmet\to\pmet$, which require the input to be non-expansive and then also guarantee the lifted predicate to be non-expansive~\cite{Forster_et_al:CSL.2023:Density}.)

\subsubsection*{Price-Function-Based Constructions}
We first discuss generalizations of the supremum side of
Kantorovich-Rubinstein duality, where distance values are computed
based on nonexpansive price functions.  For liftings and lax
extensions of functors $\set\to\set$, these constructions have already
been studied
extensively~\cite{bbkk:coalgebraic-behavioral-metrics,FuzzyLaxHemi}.
The construction for functors $\pmet\to\set$ is quite similar to the
other two (see e.g.\ \cite{KantorovichFunctors}).

\begin{defn}
  Let $\lambda$ be a predicate lifting for $F\colon\set\to\set$.
  The \emph{price-function-based lifting} of $F$ is defined as
  \begin{equation*}
    \kantsym{\lambda}(d)(s, t) = \sup \{|\lambda_X(f)(t) -
      \lambda_X(f)(s)| \mid f\colon (X,d)\to([0,1],\de) \text{ nonexpansive}\},
  \end{equation*}
  where $(X,d)$ is a pseudometric space and $s,t\in FX$.
\end{defn}

Under this definition, each $\kantsym{\lambda}$ is a pseudometric and each $Ff$ is a nonexpansive map~\cite{bbkk:coalgebraic-behavioral-metrics}, so that $(X,d) \mapsto (FX,\kantsym{\lambda})$ indeed defines a functor lifting.

The fuzzy-relational version of the price-function-based construction is given as follows:
\begin{defn}
  Let $\lambda$ be a predicate lifting for $F\colon\set\to\set$.
  The \emph{price-function-based extension} of $F$ is defined as
  \begin{equation*}
    \kantrel{\lambda}(r)(s, t) = \sup \{\lambda_Y(g)(t) \ominus
      \lambda_X(f)(s) \mid (f,g) \text{ $r$-nonexpansive}\},
  \end{equation*}
  where $X$ and $Y$ are sets, $r\colon X\relto Y$, $s\in FX$ and $t\in FY$.
  
\end{defn}

\noindent This construction indeed yields a lax extension~\cite{FuzzyLaxHemi}.

Finally, we give a construction for $\pmet$-based functors, which as mentioned before does not really differ from the one for $\set$-based functors:
\begin{defn}
  Let $\lambda$ be a predicate lifting for $F\colon\pmet\to\set$.
  The \emph{price-function-based lifting} of $F$ is defined as
  \begin{equation*}
    \kantsym{\lambda}(d)(s, t) = \sup \{|\lambda_X(f)(t) -
      \lambda_X(f)(s)| \mid f\colon (X,d)\to([0,1],\de) \text{ nonexpansive}\},
  \end{equation*}
  where $(X,d)$ is a pseudometric space and $s,t\in FX$.
\end{defn}
\begin{lem}
  $\kantsym{\lambda}(d)$ is indeed a pseudometric.
\end{lem}
\begin{proof}
  The proof is the same as for endofunctors on $\set$ (e.g.~\cite[Theorem 5.5]{bbkk:coalgebraic-behavioral-metrics}).
\end{proof}

\subsubsection*{Coupling-Based Constructions}
Next, we generalize starting from the infimum side of Kantorovich-Rubinstein duality, resulting in constructions based on categorical couplings.
Again, the constructions for functors of type $\set\to\set$ are already known~\cite{bbkk:coalgebraic-behavioral-metrics,FuzzyLaxHemi}, while the construction for functors of type $\pmet\to\set$ is new.

Even for set-based functors, the coupling-based construction needs to assume some conditions on the functor and predicate lifting.
Recall that a \emph{weak pullback} is a commuting square $f\cdot u = g\cdot v$ which has the weak universal property that for every $s$ and $t$ such that $f\cdot s = g\cdot t$ there exists a (not necessarily unique) $p$ such that $u\cdot p = s$ and $v\cdot p = t$.
A functor is said to \emph{preserve weak pullbacks} if it sends every weak pullback square $f\cdot s = g\cdot t$ to a weak pullback square $Ff\cdot Fs = Fg\cdot Ft$.

For the predicate lifting we impose the following requirements:
\begin{defn}
  A predicate lifting~$\lambda$ (of either of the types in \cref{defn:pred-lifting}) is \emph{well-behaved} if the following conditions hold:
  \begin{itemize}
    \item \emph{Monotonicity}: If $f \leq g$, then $\lambda_X(f)
      \leq \lambda_X(g)$, where the ordering on functions is
      computed pointwise. 
    \item \emph{Subadditivity}: for any $f, g, h$, we
      have $f \le g \oplus h \implies \lambda_X(f) \le \lambda_X(g) \oplus \lambda_X(h)$,
      where the (truncated) sum of two functions is calculated pointwise.
    \item \emph{Preservation of zero}: $\lambda_X(0_X) = 0_{FX}$,
      where $0_X$, $0_{FX}$ are the constant zero functions on the
      respective sets.
  \end{itemize}
\end{defn}

\begin{defn}[Coupling]
  Let $F\colon\catC\to\set$ be a functor, let $X$ and $Y$ be objects of $\catC$, and let $s\in FX$ and $t\in FY$.
  The set of \emph{couplings} $\Gamma(s,t)$ is defined as
  \begin{equation*}
    \Gamma(s,t) :=\{ c \in F(X \times Y) \mid F\pi_1(c) = s \text{ and } F\pi_2(c) = t\}.
  \end{equation*}
\end{defn}
\begin{defn}\label{def:wasserstein-fun}
  Let $F\colon\set\to\set$ be a functor that preserves weak pullbacks, and let $\lambda$ be a well-behaved predicate lifting for $F$.
  The \emph{coupling-based lifting} of $F$ to the category of pseudometric spaces is defined as the lifting that equips~$FX$ with the pseudometric~$W_\lambda(d_X)$ where
  \begin{equation*}
    W_\lambda(d_X)(s,t)= \inf \{ \lambda_{X\times X}(d_X)(c) \mid c \in \Gamma(s,t) \}.
  \end{equation*}
\end{defn}

\noindent The definition for fuzzy relations is practically identical:
\begin{defn}\label{def:wasserstein-rel}
  Let $F\colon\set\to\set$ be a functor that preserves weak pullbacks, and let $\lambda$ be a well-behaved predicate lifting for $F$.
  The \emph{coupling-based extension} of $F$ is defined on a fuzzy relation $r\colon X\relto Y$ and $s\in FX$, $t\in FY$ by
  \begin{equation*}
    W_\lambda(r)(s,t)= \inf \{ \lambda_{X\times Y}(r)(c) \mid c \in \Gamma(s,t) \}.
  \end{equation*}
\end{defn}

\noindent If both the functor and the predicate lifting satisfy the mentioned conditions, then these constructions indeed yield a pseudometric lifting and a lax extension, respectively~\cite{bbkk:coalgebraic-behavioral-metrics,FuzzyLaxHemi}.\medskip

Adapting the coupling-based construction to work with functors of type $\pmet\to\set$ requires a bit more work.
This is because a pseudometric $d$ on a set $X$ in general fails  to
be  nonexpansive as a map $(X,d)\times(X,d)\to([0,1],d_e)$ (where
$\times$ denotes the categorical product of pseudometric spaces).
For instance, the pairs $(0.1,0.4)$ and $(0.2,0.3)$ have distance $0.1$ in the product (supremum) pseudometric, while $|d_e(0.1,0.4) - d_e(0.2,0.3)| = 0.2$.
A suitable generalization of the coupling-based construction will therefore need to involve the Manhattan product $(X,d)\boxplus(X,d)$, on which $d$ is easily seen to  be nonexpansive.
For purposes of proving that the coupling-based lifting always yields a pseudometric, however, we still need to construct couplings as elements of the space $F((X,d)\times(X,d))$.
This naturally leads to the following condition on $F$:

\begin{defn}
  A functor $F\colon\pmet\to\set$ is called \emph{product-agnostic} if for any two pseudometric spaces $(X,d_1)$ and $(Y,d_2)$ the map 
  \begin{equation}\label{eqn:product-agnostic}
    F\langle\pi_1,\pi_2\rangle\colon F((X,d_1)\boxplus(Y,d_2))\to F((X,d_1)\times(Y,d_2))
  \end{equation}
  is an isomorphism, and moreover for any nonexpansive maps $f\colon (A,d_A)\to(X,d_1)\times(Y,d_2)$, $g\colon(X,d_1)\boxplus(Y,d_2)\to(B,d_B)$ and $h\colon(A,d_A)\to(B,d_B)$ such that $Ug\cdot Uf = Uh$, the following diagram commutes:
  \begin{equation*}
    \begin{tikzcd}
      {F(A,d_A)} & & & {F(B,d_B)} \\
      {F((X,d_1)\times(Y,d_2))} & & & {F((X,d_1)\boxplus(Y,d_2))}
      \arrow["Fh"', from=1-1, to=1-4]
      \arrow["Ff", from=1-1, to=2-1]
      \arrow["{(F\langle\pi_1,\pi_2\rangle)^{-1}}", from=2-1, to=2-4]
      \arrow["Fg", from=2-4, to=1-4]
    \end{tikzcd}
  \end{equation*}
\end{defn}
\noindent
The second part of the condition is motivated by the observation that, while the map $\langle\pi_1,\pi_2\rangle\colon(X,d_1)\boxplus(Y,d_2)\to(X,d_1)\times(Y,d_2)$ is nonexpansive, its inverse is not.
The property then ensures that maps that factor through a product $X\times Y$ as set maps, but whose factors cannot compose as nonexpansive maps due to the different pseudometrics used on the product, can nevertheless be factored after applying the functor in terms of the $F$-images of their factors and the isomorphism.
It is needed to establish reflexivity in the proof of \cref{lem:wasserstein-pmet} below.

% \begin{conv}
%   For ease of notation and without loss of generality, we assume that if $F$ is product-agnostic then for every $(X,d_1)$ and $(Y,d_2)$ we have $F((X,d_1)\boxplus(Y,d_2)) = F((X,d_1)\times(Y,d_2))$ and $F\langle\pi_1,\pi_2\rangle=\id_{F((X,d_1)\times(Y,d_2))}$.
% \end{conv}

% \begin{defn}
%   A functor $F\colon\pmet\to\set$ is \emph{product-agnostic} if for any two pseudometric spaces $(X,d_1)$ and $(Y,d_2)$, we have $F((X,d_1)\boxplus(Y,d_2)) = F((X,d_1)\times(Y,d_2))$ and $F\langle\pi_1,\pi_2\rangle=\id_{F((X,d_1)\times(Y,d_2))}$ where $\langle\pi_1,\pi_2\rangle\colon(X,d_1)\boxplus(Y,d_2)\to(X,d_1)\times(Y,d_2)$ (a non-expansive map).
% \end{defn}

Note that this condition works out nicely on the examples we discussed in \cref{sec:lift-and-lax}:
The Manhattan distance and the supremum distance induce the same topology on a set $X\times Y$, so that the closed subsets of $X\times Y$ and the Radon measures on $X\times Y$ are the same in either case, so that the map from \eqref{eqn:product-agnostic} is in fact an identity map.
This also implies the second part of the property, which intuitively speaking captures the fact that the action of the distribution functor on non-expansive maps (taking image measures) depends only on the underlying map, not on the metrics on the domain and codomain.

Using this terminology, the coupling-based construction is phrased as follows.

\begin{defn}\label{def:wasserstein-pmet}
  Let $F\colon\pmet\to\set$ be a product-agnostic functor that preserves weak pullbacks, and let $\lambda$ be a well-behaved predicate lifting for $F$.
  The \emph{coupling-based lifting} of~$F$ is defined as the lifting that equips $F(X,d)$ with the pseudometric $\wass{\lambda}(d)$ where
  \begin{equation*}
    \wasspmet{\lambda}(d)(s,t) =
    \inf \{ \lambda_{(X,d)\boxplus(X,d)}(d)((F\langle\pi_1,\pi_2\rangle)^{-1}(c)) \mid c \in \Gamma(s,t) \}.
  \end{equation*}
\end{defn}

\begin{lem}\label{lem:wasserstein-pmet}
  For every pseudometric space $(X,d)$, $\wasspmet{\lambda}(d)$ is indeed a pseudometric on $F(X,d)$.
  Furthermore, for every nonexpansive map $f\colon(X,d_1)\to(Y,d_2)$, the map $Ff$ is nonexpansive with respect to $\wasspmet{\lambda}(d_1)$ and $\wasspmet{\lambda}(d_2)$.
\end{lem}
\begin{proof}
  Throughout the proof, we write $\alpha = (F\langle\pi_1,\pi_2\rangle)^{-1}$.

  For reflexivity, let $t\in F(X,d)$.
  We need to show $\wasspmet{\lambda}(d)(t,t) = 0$.
  We apply the second part of the fact that $F$ is product-agnostic to the nonexpansive maps $\langle\id,\id\rangle\colon (X,d)\to(X,d)\times(X,d)$, $d\colon(X,d)\boxplus(X,d)\to([0,1],d_e)$ and $0_X\colon(X,d)\to([0,1],d_e)$, which results in the equality
  \begin{equation}\label{eqn:lem-wasserstein-pmet-1}
    Fd\cdot\alpha\cdot F\langle\id,\id\rangle = F0_X.
  \end{equation}
  Put $c = F\langle\id,\id\rangle(t)$; then $c\in\Gamma(t,t)$ and we have
  \begin{align*}
    \wasspmet{\lambda}(d)(t,t)
    &\le \lambda_{(X,d)\boxplus(X,d)}(d)(\alpha(c)) &&\myby{$c$ is a coupling}\\
    &= \lambda_{([0,1],d_e)}(\id_{[0,1]})(Fd\cdot\alpha\cdot F\langle\id,\id\rangle(t)) &&\myby{$\lambda$ natural} \\
    &= \lambda_{([0,1],d_e)}(\id_{[0,1]})(F0_X(t)) &&\myby{\ref{eqn:lem-wasserstein-pmet-1}} \\
    &= \lambda_{(X,d)}(0_X)(t) &&\myby{$\lambda$ natural}\\
    &= 0_{F(X,d)}(t) = 0. &&\myby{$\lambda$ well-behaved}
  \end{align*}
  For symmetry, let $t_1,t_2\in F(X,d)$.
  We have the following commuting square:
  \begin{equation*}
    \begin{tikzcd}
      {(X,d)\times(X,d)} & & & {(X,d)\times(X,d)} \\
      {(X,d)\boxplus(X,d)} & & & {(X,d)\boxplus(X,d)}
      \arrow["{\langle\pi_2,\pi_1\rangle}"', from=1-1, to=1-4]
      \arrow["{\langle\pi_1,\pi_2\rangle}", from=1-1, to=2-1]
      \arrow["{\langle\pi_2,\pi_1\rangle}", from=2-1, to=2-4]
      \arrow["{\langle\pi_1,\pi_2\rangle}", from=1-4, to=2-4]
    \end{tikzcd}
  \end{equation*}
  Applying $F$ to this square results in $F\langle\pi_2,\pi_1\rangle\cdot\alpha^{-1} = \alpha^{-1}\cdot F\langle\pi_2,\pi_1\rangle$.
  We also have 
  \begin{equation}\label{eqn:lem-wasserstein-pmet-2}
    \alpha\cdot F\langle\pi_2,\pi_1\rangle = \alpha\cdot F\langle\pi_2,\pi_1\rangle\cdot\alpha^{-1}\cdot\alpha = \alpha\cdot \alpha^{-1}\cdot F\langle\pi_2,\pi_1\rangle\cdot\alpha = F\langle\pi_2,\pi_1\rangle\cdot\alpha.
  \end{equation}
  % $\langle\pi_2,\pi_1\rangle\colon(X,d)\times()$ and $\langle\pi_1,\pi_2\rangle$
  Therefore, we may use the fact that $F\langle\pi_2,\pi_1\rangle$ is a bijection between $\Gamma(t_1,t_2)$ and $\Gamma(t_2,t_1)$ to conclude
  \begin{align*}
    \wasspmet{\lambda}(d)(t_2,t_1)
    &= \inf \{ \lambda_{(X,d)\boxplus(X,d)}(d)(\alpha(c)) \mid c \in \Gamma(t_2,t_1) \} \\
    &= \inf \{ \lambda_{(X,d)\boxplus(X,d)}(d)(\alpha\cdot F\langle\pi_2,\pi_1\rangle(c)) \mid c \in \Gamma(t_1,t_2) \} \\
    &= \inf \{ \lambda_{(X,d)\boxplus(X,d)}(d)(F\langle\pi_2,\pi_1\rangle\cdot\alpha(c)) \mid c \in \Gamma(t_1,t_2) \} &&\myby{\ref{eqn:lem-wasserstein-pmet-2}}\\
    &= \inf \{ \lambda_{(X,d)\boxplus(X,d)}(d\cdot\langle\pi_2,\pi_1\rangle)(\alpha(c)) \mid c \in \Gamma(t_1,t_2) \} &&\myby{$\lambda$ natural}\\
    &= \inf \{ \lambda_{(X,d)\boxplus(X,d)}(d)(\alpha(c)) \mid c \in \Gamma(t_1,t_2) \} &&\myby{$d$ symmetric}\\
    &= \wasspmet{\lambda}(d)(t_1,t_2).
  \end{align*}
  For the triangle inequality, let $t_1,t_2,t_3\in F(X,d)$.
  Let $t_{12}\in\Gamma(t_1,t_2)$ and $t_{23}\in\Gamma(t_2,t_3)$.
  It suffices to show that there exists $t_{13}\in\Gamma(t_1,t_3)$ such that
  \begin{equation}\label{eqn:lem-wasserstein-pmet-3}
    \lambda_{(X,d)\boxplus(X,d)}(d)(\alpha(t_{13})) \le
    \lambda_{(X,d)\boxplus(X,d)}(d)(\alpha(t_{12})) \oplus
    \lambda_{(X,d)\boxplus(X,d)}(d)(\alpha(t_{23})).
  \end{equation}
  We have a pullback $\pi_2\cdot\pi_{12} = \pi_1\cdot\pi_{23}\colon (X,d)\times (X,d)\times (X,d)\to (X,d)$, which is mapped under~$F$ to a weak pullback $F\pi_2\cdot F\pi_{12} = F\pi_1\cdot F\pi_{23}$ because $F$ preserves weak pullbacks.
  Viewing the couplings $t_{12}$ and $t_{23}$ as maps of type $1\to F((X,d)\times(X,d))$ results in a commuting square $F\pi_2\cdot t_{12} = F\pi_1\cdot t_{23}$, so by the weak pullback property there is a map/element $t_{123}\colon 1\to F((X,d)\times(X,d)\times(X,d))$ such that $F\pi_{23}(t_{123}) = t_{23}$ and $F\pi_{12}(t_{123}) = t_{12}$.
  Now put $t_{13} = F\langle\pi_1,\pi_3\rangle(t_{123})$; it remains to show~\eqref{eqn:lem-wasserstein-pmet-3}.

  For each pair of indices $1 \le i < j \le 3$, by an argument identical to the one from the symmetry case we have $\alpha\cdot F\langle\pi_i,\pi_j\rangle = F\langle\pi_i,\pi_j\rangle\cdot\alpha$, and therefore also
  \begin{align}
    \lambda_{(X,d)\boxplus(X,d)}(d)(\alpha(t_{ij}))
    &= \lambda_{(X,d)\boxplus(X,d)}(d)(F\langle\pi_i,\pi_j\rangle\cdot\alpha(t_{123})) \nonumber \\
    &= \lambda_{(X,d)\boxplus(X,d)\boxplus(X,d)}(d\cdot\langle\pi_i,\pi_j\rangle)(\alpha(t_{123})). \label{eqn:lem-wasserstein-pmet-4}
  \end{align}
  Write $d_{ij} = d\cdot\langle\pi_i,\pi_j\rangle$.
  Then $d_{13} \le d_{12} \oplus d_{23}$.
  As $\lambda$ is well-behaved and therefore also subadditive, we can rewrite each of the three terms in \eqref{eqn:lem-wasserstein-pmet-3} according to \eqref{eqn:lem-wasserstein-pmet-4} and finish the proof by applying subadditivity.

  Finally, we show functoriality. Let $f\colon(X,d_1)\to(Y,d_2)$ be a nonexpansive map.
  Then we have a commuting square
  \begin{equation*}
    \begin{tikzcd}
      {(X,d_1)\times(X,d_1)} & & & {(Y,d_2)\times(Y,d_2)} \\
      {(X,d_1)\boxplus(X,d_1)} & & & {(Y,d_2)\boxplus(Y,d_2)}
      \arrow["f\times f"', from=1-1, to=1-4]
      \arrow["{\langle\pi_1,\pi_2\rangle}", from=1-1, to=2-1]
      \arrow["f\boxplus f", from=2-1, to=2-4]
      \arrow["{\langle\pi_1,\pi_2\rangle}", from=1-4, to=2-4]
    \end{tikzcd}
  \end{equation*}
  Applying $F$ and using the same steps as in \eqref{eqn:lem-wasserstein-pmet-2}, we have $\alpha\cdot F(f\times f) = F(f\boxplus f)\cdot\alpha$.
  Let $t_1,t_2\in F(X,d_1)$.
  Note that every coupling of $t_1$ and $t_2$ can be made into a coupling of $Ff(t_1)$ and $Ff(t_2)$ using the map $F(f\times f)$.
  Therefore:
  \begin{align*}
    &\wasspmet{\lambda}(d_2)(Ff(t_1),Ff(t_2)) \\
    &= \inf \{ \lambda_{(Y,d_2)\boxplus(Y,d_2)}(d_2)(\alpha(c)) \mid c \in \Gamma(Ff(t_1),Ff(t_2)) \} \\
    &\le \inf \{ \lambda_{(Y,d_2)\boxplus(Y,d_2)}(d_2)(\alpha\cdot F(f\times f)(c)) \mid c \in \Gamma(t_1,t_2) \} \\
    &= \inf \{ \lambda_{(Y,d_2)\boxplus(Y,d_2)}(d_2)(F(f\boxplus f)\cdot\alpha(c)) \mid c \in \Gamma(t_1,t_2) \} \\
    &= \inf \{ \lambda_{(X,d_1)\boxplus(X,d_1)}(d_2\cdot(f\boxplus f))(\alpha(c)) \mid c \in \Gamma(t_1,t_2) \} &&\myby{$\lambda$ natural}\\
    &\le \inf \{ \lambda_{(X,d_1)\boxplus(X,d_1)}(d_1)(\alpha(c)) \mid c \in \Gamma(t_1,t_2) \} &&\myby{$\lambda$ monotone, $f$ nonexp.}\\
    &=\wasspmet{\lambda}(d_1)(t_1,t_2) &&\qedhere
  \end{align*}
\end{proof}

\begin{rem}
  It is well known that predicate liftings correspond to morphisms of
  certain type, sometimes dubbed \emph{evaluation functions}, by the
  Yoneda lemma \cite{Schroeder2007:ExpressivityCoalgModalLog}. In the
  concrete instance of $\preal$-valued predicate liftings, we have
  that natural transformations of the form
  $\lambda\colon \preal^- \To \preal^{F-}$ are in bijection with
  morphisms of type $\mathsf{ev}_\lambda\colon F\preal \to
  \preal$. Then the condition of a predicate lifting being
  well-behaved translates roughly to the corresponding evaluation
  function being well-behaved
  \cite{bbkk:coalgebraic-behavioral-metrics,bkp:up-to-behavioural-metrics-fibrations-journal,FuzzyLaxHemi}.
  % \bknote{I
  %   would say that standard and well-behaved are the same thing, for a
  %   comparison, see
  %   \cite{bkp:up-to-behavioural-metrics-fibrations-journal}. Later in
  %   the paper we use ``well-behaved''.}
  %   \jfnote{Check
  %   if roughly is needed. I thought predicate liftings had slightly
  % weaker requirements for well-behavedness?}
  % The choice between predicate
  % liftings and evaluation functions is largely a matter of preference;
  % the work below is phrased in terms of predicate liftings.
\end{rem}

\begin{rem}
  The restriction of the presentation of the price-function-based
  construction to unary predicate liftings above is purely in the
  interest of readability; generalizing to higher-arity predicate
  liftings (i.e.~ones that take several predicates as arguments)
  requires no more than additional indexing. On the coupling-based
  side, however, it is currently unclear how a generalization to
  higher-arity predicate liftings would look; we leave this point to
  future investigation.
\end{rem}

\subsubsection*{Instances}
As one would expect, the probabilistic case from which we generalized can be recovered as an instance of both types of constructions.
The various forms of the Hausdorff lifting/extension are also instances of both the price-function-based and coupling-based constructions.

\begin{expl}
  \label{expl:kant-wass-instances}
  \begin{enumerate}[wide]
    \item\label{item:powerset-instance} Let $\pfun\colon\set\to\set$ be the powerset functor, and let $\lambda \colon \preal^X \to \preal^{\pfun -}$ be the natural transformation whose components calculate suprema of direct images:
      For $A \subseteq X$ and $f\in \preal^X$ we define $\lambda_X(f)(A) = \sup f[A]$.
      Then $\dh(d) = \kantsym{\lambda}(d) = \wass{\lambda}(d)$ for every pseudometric $d$~\cite{bbkk:coalgebraic-behavioral-metrics}.
      Additionally, let $\mu$ be the predicate lifting that calculates infima, i.e.\ $\mu_X(f)(A) = \inf f[A]$.
      Then $\dh(r) = \max(\kantrel{\lambda}(d),\kantrel{\mu}(d)) = \wass{\lambda}(d)$ for every fuzzy relation~$r$.
    \item \label{item:distribution-instance} Let $\dfun\colon\set\to\set$ be the discrete distribution functor, and let $\lambda$ be the predicate lifting calculating expected values:
      For $\mu \in \dfun X$ and $f\in \preal^X$ we have
      $\lambda_X(f)(\mu) = \expect{\mu}{f}$.
      Then $\dk(d) = \kantsym{\lambda}(d) = \wass{\lambda}(d)$ for every pseudometric $d$,
      and $\dkrel(r) = \kantrel{\lambda}(r) = \wass{\lambda}(r)$ for every fuzzy relation $r$.
    \item Let $\pfun^\mathsf{c}\colon\pmet\to\set$ be the functor that maps each pseudometric space $(X,d)$ to the set of closed subsets of $X$ in the topology induced by the pseudometric $d$, and let $\lambda$ be as in \cref{item:powerset-instance}.
      We have $\dh(d) = \kantsym{\lambda}(d) = \wasspmet{\lambda}(d)$ for each pseudometric $d$.
    \item\label{item:radon} Let $\dfun^\mathsf{Rad}\colon\pmet\to\set$ be the functor that maps each pseudometric space $(X,d)$ to the set of Radon measures on $(X,d)$.
      Let $\lambda$ be Lebesgue integration (note that every nonexpansive function is Borel-measurable).
      Because Kantorovich-Rubinstein duality holds for arbitrary Radon measures~\cite{Kellerer85}, we have $\dk(d) = \kantsym{\lambda}(d) = \wasspmet{\lambda}(d)$ for each pseudometric~$d$.
  \end{enumerate}
\end{expl}

\section{Duality}
\label{sec:duality}

Our present interest is in  the question of \emph{generalized Kantorovich-Rubinstein duality}, that is, whether the categorical price-function-based and coupling-based constructions coincide, for the same choice of predicate lifting $\lambda$.
As we have seen in \cref{expl:kant-wass-instances}, duality in this sense holds for both the powerset with the supremum modality and the distribution functor with the expected value modality, for both kinds of functor liftings and for lax extensions.
In the case of lax extensions of the powerset functor, however, we needed to additionally involve the infimum modality, which is the \emph{dual} of the supremum modality -- essentially, this is due to the fact that the definition of $\kantrel{}$ features directed differences instead of absolute differences.
\begin{defn}[Dual predicate lifting]
  Let $\lambda\colon\set\to\set$ be a functor, and let $\lambda$ be a predicate lifting for $F$.
  The \emph{dual predicate lifting} $\overline{\lambda}$ is given by $\overline{\lambda}_X(f)(t) = 1 - \lambda_X(1-f)(t)$ for each $f\in[0,1]^X$ and $t\in FX$.
\end{defn}

\begin{conv}
  If $\Lambda$ is a set of predicate liftings, we write
  \begin{equation*}
    \kantsym{\Lambda} = \sup_{\lambda\in\Lambda} \kantsym{\lambda}
    \qquad\text{and}\qquad
    \kantrel{\Lambda} = \sup_{\lambda\in\Lambda} \kantrel{\lambda}.
  \end{equation*}
\end{conv}

\noindent
In the fuzzy-relational setting, then, we say that \emph{generalized Kantorovich-Rubinstein duality} holds if for every fuzzy relation $r$,
\begin{equation}\label{eq:duality-rel}
  \kantrel{\{\lambda,\overline{\lambda}\}}(r) = \wass{\lambda}(r).
\end{equation}

\noindent
In this terminology, we have generalized Kantorovich duality both in the case of the Hausdorff extension and the Kantorovich-Rubinstein extension.
In the case of the latter, note that the expected value modality is equal to its own dual.

Generalized duality may fail, as demonstrated by the case of $p$-Wasserstein distance that we discuss later.
Nevertheless, we always have one of the two inequalities, namely that the price-function-based construction is below the coupling-based construction, both for pseudometric liftings~\cite[Theorem 5.27]{bbkk:coalgebraic-behavioral-metrics} and lax extensions~\cite[Lemma 5.22]{PaulWildThesis} of functors of type $\set\to\set$
For pseudometric liftings of functors of type $\pmet\to\set$ this is also true:

\begin{lem}\label{lem:kant-le-wass}
  Let $F\colon\pmet\to\set$ be a product-agnostic functor, and let $\lambda$ be a well-behaved predicate lifting for $F$.
  Then 
  $\kantsym{\lambda}(d) \le \wasspmet{\lambda}(d)$
  for every pseudometric $d$.
\end{lem}
\begin{proof}
  Let $(X,d)$ be a pseudometric space, and let $t_1,t_2\in F(X,d)$.
  We will need to distinguish between projections from the two kinds of product of $(X,d)$ with itself, so for $i=1,2$, write $\pi^\boxplus_i\colon(X,d)\boxplus(X,d)\to(X,d)$ and $\pi^\times_i\colon(X,d)\times(X,d)\to(X,d)$ for the projection maps (all of which are nonexpansive).

  Let $c\in\Gamma(t_1,t_2)$ and let $f\colon(X,d)\to([0,1],d_e)$ be nonexpansive.
  We need to show that
  \begin{equation}\label{eqn:kant-le-wass}
    \lambda_{(X,d)\boxplus(X,d)}(d)(\alpha(c)) \ge \lambda_{(X,d)}(f)(t_2) \ominus \lambda_{(X,d)}(f)(t_1),
  \end{equation}
  where $\alpha = (F\langle\pi^\boxplus_1,\pi^\boxplus_2\rangle)^{-1}\colon F((X,d)\times(X,d))\to F((X,d)\boxplus(X,d))$.

  We have $F\pi^\boxplus_i\cdot\alpha = F\pi^\times_i$ for $i=1,2$, and therefore
  \begin{equation*}
    \lambda_{(X,d)\boxplus(X,d)}(f\cdot\pi^\boxplus_i)(\alpha(c))
    = \lambda_{(X,d)}(f)(F\pi^\boxplus_i\cdot\alpha(c))
    = \lambda_{(X,d)}(f)(F\pi^\times_i(c))
    = \lambda_{(X,d)}(f)(t_i).
  \end{equation*}
  Rewriting with these equations in \eqref{eqn:kant-le-wass} and rearranging the terms, our new goal is
  \begin{equation*}
    \lambda_{(X,d)\boxplus(X,d)}(f\cdot\pi^\boxplus_2)(\alpha(c)) \le
    \lambda_{(X,d)\boxplus(X,d)}(f\cdot\pi^\boxplus_1)(\alpha(c)) \oplus
    \lambda_{(X,d)\boxplus(X,d)}(d)(\alpha(c)).
  \end{equation*}
  This follows by subadditivity of $\lambda$, applied to the inequality
  $f\cdot\pi^\boxplus_2 \le f\cdot\pi^\boxplus_1 \oplus d$, which holds because $f$ is nonexpansive.
\end{proof}

\noindent
Proofs of generalized Kantorovich-Rubinstein duality therefore generally hinge on establishing the other inequality.

\subsubsection*{Expressive Logics}
The price-function-based lifting relates closely to
\emph{characteristic} quantitative modal logics, i.e.~logics whose
induced logical distance coincides with behavioural distance in the
sense recalled above. We briefly discuss the shape that the
corresponding definitions and results take when instantiated to our
current setting. Fix a functor $F\colon\pmet\to\set$. Given a choice
of a set~$\Lambda$ of predicate liftings for~$F$, one defines the set
$\FLang(\Lambda)$ of \emph{$\Lambda$-formulae}~$\phi,\psi,\dots$ in
the associated \emph{(non-expansive) quantitative coalgebraic modal
  logic} by the grammar
\begin{equation*}
  \FLang(\Lambda)\owns\phi,\psi::=\bot\mid\top\mid\phi\oplus a\mid\phi\ominus a\mid\phi\land\psi\mid\phi\lor\psi\mid\lambda\phi\qquad(\lambda\in\Lambda,a\in[0,1]\cap\Rat)
\end{equation*}
(where we thus abuse predicate liftings as syntactic modalities). We
write $\overline F\colon\pmet\to\pmet$ for the price-function-based
lifting of~$F$ induced by~$\Lambda$. The semantics of a formula~$\phi$
over an $\overline F$-coalgebra $\gamma\colon C\to \overline FC$ is
then given in terms of a non-expansive evaluation map
$\sem{\phi}\colon C\to [0,1]$, defined recursively by
\begin{align*}
  \sem{\bot}(c) & = 0 & \sem{\top}(c) & = 1\\
  \sem{\phi\land\psi}(c) & =\min(\sem{\phi}(c),\sem{\psi}(c)) 
                      & \sem{\phi\lor\psi}(c) & = \max(\sem{\phi}(c),\sem{\psi}(c)) \\
  \sem{\phi\oplus a}(c) & = \sem{\phi}(c)\oplus a &
                                                  \sem{\phi\ominus a}(c) & = \sem{\phi}(c)\ominus a \\
  \sem{\lambda\phi}(c) & = \lambda(\sem{\phi})(\gamma(c))
\end{align*}
Here, non-expansiveness is an invariant of the recursive definition,
with the modal step taken care of by the fact that non-expansive
predicates on~$C$ are lifted to non-expansive predicates
on~$\overline FC$ by construction of~$\overline F$. We then define the
\emph{logical distance} of states~$c\in C$, $d\in D$ in $F$-coalgebras
$\gamma\colon C\to FC$, $\delta\colon D\to FD$, respectively, as the
supremum
\begin{equation*}
  \sup_{\phi\in\FLang(\Lambda)}\sem{\phi}(d)-\sem{\phi}(c)
\end{equation*}
of differences in truth values between~$c$ and~$d$, taken over all
$\Lambda$-formulae~$\phi$. By non-expansiveness of formula evaluation,
this logical distance is always below the behavioural distance induced
by~$\overline F$. The converse inequality requires additional
assumptions. We work with the following additional property of
predicate liftings:
\begin{defn}
  A predicate lifting $\lambda$ is \emph{non-expansive} if all
  components~$\lambda_X$ are non-expansive as maps
  $[0,1]^X\to[0,1]^{FX}$ when function spaces $[0,1]^{(-)}$ are
  equipped with the supremum distance.
\end{defn}
% \noindent Also, we will want to generalize beyond functor liftings, as
% some of our results will concern functors that fail to be liftings of
% set functors. In this view, the Kantorovich construction of liftings
% becomes a property of functors on metric spaces:
% \begin{defn}
%   Let~$\Lambda$ be a set of predicate liftings for a functor
%   $F\colon\pmet\to\pmet$\lsnote{This needs to be defined; it
%     corresponds to the previous notion of `preserves
%     non-expansiveness'}. Then~$F$ is \emph{$\Lambda$-Kantorovich} if
%   for every pseudometric space~$X$, the cone of all maps
%   $\lambda(f)\colon FX\nexp[0,1]$, for~$\lambda\in\Lambda$ and
%   $f\colon X\nexp[0,1]$, is initial\lsnote{this needs to be defined, and notation introduced};
%   explicitly:
%   \begin{equation*}
%     d_{FX}(s,t)=\sup\{|\lambda(f)(t)-\lambda(f)(s)|\mid \lambda\in\Lambda, f\colon X\nexp [0,1]\}.
%   \end{equation*}
% \end{defn}
\noindent Moreover, we need a smallness condition on the functor
generalizing the standard condition of finitariness imposed in the
two-valued
setting~\cite{Pattinson04,Schroeder2007:ExpressivityCoalgModalLog}
(and already in the standard Hennessy-Milner theorem, which holds
only for finitely branching labelled transition systems). The
corresponding metric notion has originally been introduced in the
context of lax extensions of set functors~\cite{FuzzyLaxHemi}. We
rephrase it for functors on pseudometric spaces, following later
work~\cite{Forster_et_al:CSL.2023:Density} but retaining the original
terminology:
\begin{defn}
  We say that $\overline F\colon\pmet\to\pmet$ is \emph{finitarily separable} if
  for every pseudometric space~$X$, the subset
  $\{\overline Fi(t)\mid t\in \overline FX_0, X_0\subseteq X\text{ finite},i\colon
  X_0\into X\}$ is dense in~$\overline FX$.
\end{defn}
\noindent For instance, the standard Kantorovich-Rubinstein lifting of the Radon-measure functor
$\dfun^\mathsf{Rad}$
(\cref{expl:kant-wass-instances}.\ref{item:radon}) is finitarily
separable. The relevant expressiveness criterion then reads as
follows~\cite[Corollary~5.10]{KantorovichFunctors}, in the notation assembled so far:
\begin{thm}[Coalgebraic quantitative Hennessy-Milner theorem]\label{thm:hm}
  Suppose that~$\overline F$ is finitarily separable, and that all
  predicate liftings in~$\Lambda$ are non-expansive. Then the logical
  distance induced by~$\Lambda$ coincides with behavioural distance on
  $\overline F$-coalgebras.
\end{thm}
\noindent This means that high behavioural distance can be witnessed
by \emph{distinguishing formulae}: Whenever states $c,d$ in
$\overline F$-coalgebras have distance more than~$\epsilon$, then
(under the conditions of \cref{thm:hm}) one can always find a
$\Lambda$-formula~$\phi$ such that
$\sem{\phi}(d)-\sem{\phi}(c)>\epsilon$.  Hence, the duality can be
used very fruitfully: % in cases where the
% Wasserstein presentation can be computed more efficiently than the
% Kantorovich presentation
One can use the coupling-based view to witness upper bounds on
behavioural distance (by prefixpoints of the relevant functional), and
the price-function-based view to witness lower bounds by means of
distinguishing
formulae.  % For instance, this is the case for the original Kantorovich
% metric on distributions or also the L\'evy-Prokhorov distance studied
% in this paper.

We briefly recall how the classical example of probabilistic
transition systems fits into this framework:
\begin{expl}
  To avoid triviality of behavioural equivalence, we slightly modify
  \cref{expl:kant-wass-instances}.\ref{item:radon}, replacing the
  Radon measure functor~$\dfun^\mathsf{Rad}$ with
  $F=1+\dfun^\mathsf{Rad}$, understanding $\star\in 1$ as
  termination. We use the predicate lifting~$\lambda$ that acts as
  Lebesgue integration on~$\dfun^\mathsf{Rad}$ and maps~$\star$ to~$0$
  (the dual predicate lifting~$\overline\lambda$ behaves the same way
  on~$\dfun^\mathsf{Rad}$ and maps~$\star$ to~$1$). This setup
  satisfies generalized Kantorovich-Rubinstein duality as
  per~\eqref{eq:duality-rel}: Both the coupling-based construction and
  the price-function-based construction induce a pseudometric on~$FX$
  that gives termination~$\star$ maximal distance (i.e.~distance~$1$)
  from all other elements, and on~$\dfun^\mathsf{Rad}$ is just
  standard Kantorovich-Rubinstein distance. The logic we obtain is a
  slight simplification of van Breugel and Worrell's quantitative
  modal logic for probabilistic transition
  systems~\cite{bw:behavioural-pseudometric} (which additionally
  features labelled actions and puts termination under the
  distribution functor). The conditions of \cref{thm:hm} are easily
  checked, so we obtain that behavioural distance on probabilistic
  transition systems, in the sense of $F$-coalgebras, coincides with
  logical distance in non-expansive quantitative coalgebraic modal
  logic with the modalities $\lambda,\overline\lambda$.
\end{expl}

\subsubsection*{The $p$-Wasserstein Distance}
% \label{sec:p-wasserstein}

There also exists a parametrized version of the Kantorovich-Rubinstein
distance, called the \emph{$p$-Wasserstein distance} for some real parameter ${p \ge 1}$.
In categorical terms, it is given by the predicate lifting $\lambda_p(f)(\mu) = (\expect{\mu}{f^p})^\frac{1}{p}$ (where $f^p$ takes the $p$-th power pointwise), that is:
\begin{equation*}
  \wass{\lambda_p}(d)(\mu,\nu) = \inf \{ (\expect{\rho}{d^p})^{\frac{1}{p}} \mid \rho \text{ is a coupling of $\mu$ and $\nu$} \}.
\end{equation*}
For $p=1$ this is just the usual Kantorovich-Rubinstein distance and duality holds.
If $p>1$, however, then duality may fail, and the corresponding price-function-based construction~$\kantsym{\lambda_p}$ may be strictly below $\wass{\lambda_p}$.
This gives witness to the general idea from transportation theory that $p=1$ constitutes a special case among the family of Wasserstein distances:

\begin{expl}
  Let $p = 2$.
  Let $(X,d)$ be a two-element discrete metric space, that is $X = \{0,1\}$ and $d(0,1) = d(1,0) = 1$, and let $\mu = \frac{2}{3}\cdot 0 + \frac{1}{3}\cdot 1$ and $\nu = \frac{1}{3}\cdot 0 + \frac{2}{3}\cdot 1$. 
  Then we have $\kantsym{\lambda_p}(d)(\mu,\nu) \le \frac{1}{3} < \frac{1}{\sqrt{3}} = \wass{\lambda_p}(d)(\mu,\nu)$.
\end{expl}
\begin{proof}
  We begin by showing the last equality.
  We note that $d^2 = d$, so that
  \begin{equation*}
    \wass{\lambda_2}(d)(\mu,\nu) = (\dk(d)(\mu,\nu))^\frac{1}{2} = (\tfrac{1}{3})^\frac{1}{2} = \tfrac{1}{\sqrt{3}},
  \end{equation*}
  where in the first step we used that $(-)^\frac{1}{2}$ is monotone and continuous.

  For the first inequality, let $f\colon X\to[0,1]$ be nonexpansive and put $a = f(0)$ and $b = f(1)$.
  We show that $\lambda_2(f)(\nu) - \lambda_2(f)(\mu) \le \frac{1}{3}$;
  the proof that $\lambda_2(f)(\mu) - \lambda_2(f)(\nu) \le \frac{1}{3}$ is analogous.
  We may assume wlog.\ that $a < b$, as otherwise the left hand side of our target inequality is nonpositive.
  Define $g\colon[0,1]\to\real$ via $g(t) = \sqrt{(1-t)a^2 + tb^2}$.
  Then we have
  $g(0) = a$,
  $g(\frac{1}{3}) = \lambda_2(f)(\mu)$,
  $g(\frac{2}{3}) = \lambda_2(f)(\nu)$ and
  $g(1) = b$.
  % Be careful here if negative values are also allowed.
  As the function $g$ is concave, we also have
  \begin{multline}\label{eq:g-concave}
    g(0) + g(1)
    = (\tfrac{1}{3}\cdot g(0) + \tfrac{2}{3}\cdot g(1)) + (\tfrac{2}{3}\cdot g(0) + \tfrac{1}{3}\cdot g(1)) \\
    \le g(\tfrac{1}{3}\cdot 0 + \tfrac{2}{3}\cdot 1) + g(\tfrac{2}{3}\cdot 0 + \tfrac{1}{3}\cdot 1)
    = g(\tfrac{1}{3}) + g(\tfrac{2}{3}).
  \end{multline}
  Additionally, nonexpansiveness of $f$ implies that
  \begin{equation}\label{eq:p-wass-helper}
    b^2 - a^2 = (b-a)(a+b) = (g(1)-g(0))(a+b) \le (1-0)(a+b) = a+b.
  \end{equation}
  Therefore we have:
  \begin{align*}
    (g(\tfrac{2}{3}) - g(\tfrac{1}{3}))\cdot (g(\tfrac{2}{3}) + g(\tfrac{1}{3}))
    &= g(\tfrac{2}{3})^2 - g(\tfrac{1}{3})^2 \\
    &= (\tfrac{1}{3}a^2 + \tfrac{2}{3}b^2) - (\tfrac{2}{3}a^2 + \tfrac{1}{3}b^2) \\
    &= \tfrac{1}{3}(b^2 - a^2) \\
    &\le \tfrac{1}{3}(g(0)+g(1)) &&\myby{\ref{eq:p-wass-helper}} \\
    &\le \tfrac{1}{3}(g(\tfrac{2}{3}) + g(\tfrac{1}{3})). &&\myby{\ref{eq:g-concave}}
  \end{align*}
  Our earlier assumption that $a < b$ implies that $g(\tfrac{2}{3}) + g(\tfrac{1}{3})$ is positive, so we can divide by it on both sides, which results in the claimed inequality.
\end{proof}

\section{L\'evy-Prokhorov Distance}
\label{sec:levy-prokhorov}

The Lévy-Prokhorov distance provides an alternative to the
Kantorovich-Rubinstein distance when it comes to measuring the distance between probability distributions.
% Even though it is usually presented as a metric on some space of probability distributions, the construction also more generally applies to fuzzy relations.
If $(X,d)$ is a pseudometric space, and $\mu,\nu\in\dfun X$ are (discrete) probability distributions, then we define
\begin{equation*}
  \label{eq:lp-definition}
  \dlp(d)(\mu,\nu) = \inf \{
    \epsilon \mid \forall A\subseteq X.\; \mu(A) \le \nu(A^d_\epsilon) + \epsilon
  \},
\end{equation*}
where $A^d_\epsilon = \{y \in X \mid \inf_{x\in A} d(x,y) \le \epsilon \}$.
The definition of the Lévy-Prokhorov distance sometimes includes the
mirrored condition $\forall B\subseteq X.\; \nu(B) \le
\mu(B^d_\epsilon) + \epsilon$, but this second clause is redundant and
does not actually change the induced pseudometric.

\begin{lem}\label{lem:lp-mirrored-clause-redundant}
  We have
  \begin{equation*}
    \dlp(d)(\mu,\nu) = \inf \{
      \epsilon \mid \forall A\subseteq X.\; \mu(A) \le \nu(A^d_\epsilon) + \epsilon \land \forall B\subseteq X.\; \nu(B) \le \mu(B^d_\epsilon) + \epsilon
      \}.
  \end{equation*}
\end{lem}
\begin{proof}
  Let $B\subseteq X$ and let $\epsilon \ge 0$, and assume that for all $A\subseteq X$ we have $\mu(A) \le \nu(A^d_\epsilon) + \epsilon$.
  It suffices to show that $\nu(B) \le \mu(B^d_\epsilon) + \epsilon$.
  Put $A = X \setminus B^d_\epsilon$.
  Then we have $A^d_\epsilon \subseteq X \setminus B$, so that
  \begin{equation*}
    1 - \mu(B^d_\epsilon)
    = \mu(A)
    \le \nu(A^d_\epsilon) + \epsilon
    \le \nu(X \setminus B) + \epsilon
    = 1 - \nu(B) + \epsilon,
  \end{equation*}
  and the claimed inequality follows by simplifying and rearranging.
\end{proof}

\noindent
The Lévy-Prokhorov distance has recently been investigated by Desharnais and Sokolova~\cite{DesharnaisSokolova26}, who prove that it is a functor lifting, but not a monad lifting, and that it characterizes the notion of $\epsilon$-bisimulation~\cite{DesharnaisEA08} when viewed as a lax extension.

The L\'evy-Prokhorov distance admits a representation based on couplings (cf.~\cref{sec:constructions}), or equivalently in terms of pairs of (not necessarily independent) random variables that are distributed according to the given distributions, known as the \emph{Ky Fan metric}~\cite{DudleyRealAnalysis}.
The predicate lifting $\lambda$ underlying this representation is given by
\begin{equation}
  \label{eq:lp-pred-lifting-1}
  \lambda_X(f)(\mu) = \inf \{\epsilon \ge 0 \mid \mu(\{x\in X \mid
    f(x) > \epsilon\}) \le \epsilon\},
\end{equation}
and using this predicate lifting we have $\dlp(d) = \wass{\lambda}(d)$ for every pseudometric~$d$, explicitly:
\begin{equation*}
  \dlp(d)(\mu,\nu) = \inf \{ \inf \{\epsilon \ge 0 \mid \rho(\{(x,y)\in X\times Y \mid
    d(x,y) > \epsilon\}) \le \epsilon\} \mid \rho\in\Gamma(\mu,\nu) \}.
\end{equation*}

\noindent
The predicate lifting from~\eqref{eq:lp-pred-lifting-1} has independently been used under the name `\emph{generally}' in work on fuzzy description logics~\cite{SchroderPattinson11}, and it admits a number of equivalent representations.
Intuitively, all of these representations amount to the statement that the value $\lambda_X(f)(\mu)$ is high if the value of $f$ is high with high probability when sampling according to the distribution~$\mu$.
\begin{lem}\label{lem:lp-pred-lifting}
  Let $X$ be a set, let $f\colon X\to[0,1]$, and let $\mu\in\dfun X$.
  Then we have:
  \begin{enumerate}
    \item $\lambda_X(f)(\mu) = \inf_{\epsilon \ge 0} \max (\mu(\{x\in X\mid
    f(x) > \epsilon\}), \epsilon)$ \label{item:lp-pred-lifting-2}
    \item $\lambda_X(f)(\mu) = \sup_{\epsilon \ge 0} \min (\mu(\{x\in X\mid f(x) > \epsilon\}), \epsilon)$ \label{item:lp-pred-lifting-3}
    \item $\lambda_X(f)(\mu) = \sup \{\epsilon \ge 0 \mid \mu(\{x\in X\mid f(x) > \epsilon\}) \ge \epsilon \}$ \label{item:lp-pred-lifting-4}
  \end{enumerate}
  All these identities, and also~\eqref{eq:lp-pred-lifting-1}, remain true if $f(x) > \epsilon$ is replaced by $f(x) \ge \epsilon$.
\end{lem}
\begin{proof}
  % We first show \cref{item:lp-pred-lifting-2}.
  % For `$\ge$', let $\epsilon>0$ such that
  % $\mu(\{x\in X \mid f(x) > \epsilon\}) \le \epsilon$. Then
  % $\max (\mu(\{x\in X\mid f(x) > \epsilon\}), \epsilon)=\epsilon$,
  % so~$\epsilon$ occurs also in the second infimum. For `$\le$', let
  % $\epsilon>0$. We distinguish cases over the maximum
  % $m=\max (\mu(\{x\in X\mid f(x) > \epsilon\}), \epsilon)$. If
  % $m=\epsilon$, then
  % $\mu(\{x\in X \mid f(x) > \epsilon\}) \le \epsilon$, so $m=\epsilon$
  % occurs also in the first infimum. Otherwise,
  % $m=\mu(\{x\in X \mid f(x) > \epsilon\}) \ge \epsilon$, so
  % \begin{equation*}
  %   \mu(\{x\in X \mid f(x) > m\})\le \mu(\{x\in X \mid f(x) >
  %   \epsilon\}) = m,
  % \end{equation*}
  % so~$m$ occurs also in the first infimum.

  All four representations of $\lambda$ are different ways of expressing the position of the unique crossing point between the monotone function $\epsilon\mapsto\epsilon$ and the anti-monotone function $\epsilon\mapsto\mu(\{x\in X\mid f(x) > \epsilon\})$.
  The latter function is stepwise constant, with the constant parts being half-open intervals that are closed on the left.
  If we replace $f(x) > \epsilon$ by $f(x) \ge \epsilon$, then the the half-open intervals are instead closed on the right, so $y$-coordinate of the crossing point may change, but its $x$-coordinate remains the same.
\end{proof}

\noindent
It follows from \cref{lem:lp-pred-lifting} that $\lambda$ is \emph{self-dual} (equal to its own dual):
\begin{lem}
  \label{lem:lp-pred-lifting-self-dual}
  For every $X$, $f\in[0,1]^X$ and $\mu\in\dfun X$ we have $\lambda_X(f)(\mu) = 1 - \lambda_X(1-f)(\mu)$.
\end{lem}
\begin{proof}
  As the values of $f$ are bounded within $[0,1]$, we can also restrict any suprema and infima in the following to be taken inside $[0,1]$.
  With this observation in mind, we have
  \begin{align*}
    1 - \lambda_X(1-f)(\mu)
    &= 1 - \inf_{0 \le \epsilon \le 1} \max(\mu(\{x\in X \mid 1-f(x) > \epsilon\}),\epsilon) \\
    &= \sup_{0 \le \epsilon \le 1} \min(1-\mu(\{x\in X \mid 1-f(x) > \epsilon\}),1-\epsilon) \\
    &= \sup_{0 \le \epsilon \le 1} \min(\mu(\{x\in X \mid 1-f(x) \le \epsilon\}),1-\epsilon) \\
    &= \sup_{0 \le \epsilon \le 1} \min(\mu(\{x\in X \mid f(x) \ge 1-\epsilon\}),1-\epsilon) \\
    &= \sup_{0 \le \epsilon \le 1} \min(\mu(\{x\in X \mid f(x) \ge \epsilon\}),\epsilon) = \lambda_X(f)(\mu).\qedhere
  \end{align*}
\end{proof}

% \pwnote{Possibly add proof that every $\lambda(f)$ is quasiconvex; though we don't use this fact anywhere.}

\noindent
The coupling-based representation above is justified by the fact that the predicate lifting is well-behaved:

\begin{lem}
  \label{lem:lp-pred-lifting-wb}
  The predicate lifting $\lambda$ as per~\eqref{eq:lp-pred-lifting-1}
  is well-behaved.
\end{lem}
\begin{proof}
  Monotonicity is clear from the definition.
  For preservation of the zero function, note that $\mu(\{x \in X \mid 0_X(x) > \epsilon\}) = \mu(\emptyset) = 0$ for every $\epsilon \ge 0$, so that $\lambda_X(0_X)(\mu) = 0$ as required.
  For subadditivity, let $f,g,h\colon X\to[0,1]$ and $\mu\in\dfun X$ such that $f \le g\oplus h$.
  We have
  \begin{align*}
    &\lambda_X(f)(\mu) \\
    &= \inf_{\epsilon\ge 0} \max(\mu(\{x\in X \mid f(x) > \epsilon\}),\epsilon) \\
    &\le \inf_{\epsilon\ge 0} \max(\mu(\{x\in X \mid g(x) \oplus h(x) > \epsilon\}),\epsilon) \\
    &= \inf_{\epsilon_1,\epsilon_2\ge 0} \max(\mu(\{x\in X \mid g(x) \oplus h(x) > \epsilon_1+\epsilon_2\}),\epsilon_1+\epsilon_2) \\
    &\le \inf_{\epsilon_1,\epsilon_2\ge 0} \max(\mu(\{x\in X \mid g(x) > \epsilon_1\}) + \mu(\{x\in X \mid h(x) > \epsilon_2\}),\epsilon_1+\epsilon_2) \\
    &\le \inf_{\epsilon_1,\epsilon_2\ge 0} \max(\mu(\{x\in X \mid g(x) > \epsilon_1\}),\epsilon_1) + \max(\mu(\{x\in X \mid h(x) > \epsilon_2\}),\epsilon_2) \\
    &= \lambda_X(g)(\mu) + \lambda_X(h)(\mu),
  \end{align*}
  where in the second inequality we used that $g(x) \oplus h(x) > \epsilon_1 + \epsilon_2$ implies that $g(x) > \epsilon_1$ or $h(x) > \epsilon_2$.
  Because we also have $\lambda_X(f)(\mu) \le 1$ already from the definition of $\lambda$, this shows that
  $\lambda_X(f)(\mu) \le \lambda_X(g)(\mu) \oplus \lambda_X(h)(\mu)$.
\end{proof}

% \subsection{Duality}

\noindent
Next, we show that the Lévy-Prokhorov distance admits a dual
representation using the same predicate lifting $\lambda$, thus establishing an instance of generalized Kantorovich-Rubinstein duality.
We prove this duality in
the setting of fuzzy lax extensions, and then obtain the duality between pseudometric liftings from that.

Out of the two representations of the Lévy-Prokhorov distance discussed earlier, the second, being based on couplings, readily generalizes to fuzzy relations.
Therefore, we define the \emph{relational Lévy-Prokhorov lifting} $\dlp$ to be the assignment that maps each fuzzy relation $r\colon X\relto Y$ to $\dlp(r) = \wass{\lambda}(r)\colon \dfun X\relto\dfun Y$.
Our goal is to show that we also have $\dlp(r) = \kantrel{\lambda}(r)$ for every fuzzy relation $r$, noting that we may omit the dual predicate lifting by \cref{lem:lp-pred-lifting-self-dual}.

We discussed in~\cref{sec:duality} that the inequality $\kantrel{\lambda}\le \wass{\lambda}$ follows from the general theory of coupling-based and price-function-based liftings, so that it suffices to prove the inequality $\kantrel{\lambda}\ge\wass{\lambda}$.
In the proof of the classical Kantorovich-Rubinstein duality (e.g~\cite[Theorem 5.10]{v:optimal-transport}), this direction amounts to, given an optimal transport plan in the shape of a coupling between the distributions at hand, constructing two price functions that correspond to the optimal cost, in the sense that they form a nonexpansive pair that witnesses the supremum in the definition of the relational price-function-based lifting $K_\expectSymb$.
In our proof of Lévy-Prokhorov duality we use a similar approach, which means that we should first understand how to phrase computation of the Lévy-Prokhorov distance in terms of a transport problem.

Let $r\colon X\relto Y$ and let $\mu\in\dfun X$ and $\nu\in\dfun Y$.
The coupling-based representation $\dlp(r)(\mu,\nu) = \wass{\lambda}(r)(\mu,\nu)$ can be rewritten for this purpose.
For $\epsilon\ge 0$, define $r^\epsilon(x,y) = 0$ if $r(x,y) < \epsilon$ and $r^\epsilon(x,y) = 1$ otherwise.
Then $\rho(\{(x,y) \mid r(x,y) \ge \epsilon\}) = \expect{\rho}{r^\epsilon}$ for every $\rho\in\dfun(X\times Y)$.
We may now use the representation of $\lambda$ from~\cref{lem:lp-pred-lifting}.\ref{item:lp-pred-lifting-2}, replacing the strict inequality with a non-strict one, and then swap the infimum over couplings inside to obtain
\begin{align*}
  \wass{\lambda}(r)(\mu,\nu)
  = \inf \{ \inf_{\epsilon \ge 0} \max(\epsilon,\expect{\rho}{r^\epsilon}) \mid \rho\in\Gamma(\mu,\nu) \}
  = \inf_{\epsilon \ge 0} \max(\epsilon,\wass{\expectSymb}(r^\epsilon)(\mu,\nu)).
\end{align*}
This means that Lévy-Prokhorov distance is determined by the solutions to the transport problems for the $r^\epsilon$.
As each such $r^\epsilon$ is a crisp relation (i.e.~only has $0$ and $1$ entries), the optimal price functions can be made crisp as well:
\begin{lem}\label{lem:01-price-functions}
  Let $r\colon X\times Y\to \{0,1\}$ be a crisp relation and let $\mu\in\dfun X$, $\nu\in\dfun Y$.
  Then there exist functions $f\colon X\to\{0,1\}$ and $g\colon Y\to\{0,1\}$ such that $(f,g)$ is an $r$-nonexpansive pair and $\expect{\nu}{g} - \expect{\mu}{f} \ge W_{\mathbb{E}}(r)(\mu,\nu)$.
\end{lem}
\begin{proof}
  % \pwnote{lower semicontinuity}
  We use~\cite[Theorem 5.10(ii)]{v:optimal-transport}, which guarantees that there exists an optimal coupling $\rho\in\dfun(X\times Y)$, as well as price functions $f\colon X\to\real$, $g\colon\ Y\to\real$ such that $(f,g)$ is $r$-nonexpansive, $\expect{\nu}{g} - \expect{\mu}{f} = \expect{\rho}{r}$, and moreover, if $U = \{(x,y) \mid g(y)-f(x) = r(x,y)\}$, then $\rho(U) = 1$ and the set $U$ is \emph{$r$-cyclically monotone}, meaning that for any $(x_1,y_1),\dots,(x_n,y_n)\in U$ we have
  \begin{equation*}
    \sum_{i=1}^n r(x_i,y_i) \le \sum_{i=1}^n r(x_i,y_{i+1}),
  \end{equation*}
  where $y_{n+1} = y_1$.
  Consider the graph on $X+Y$ with edges given by $U$.
  We can modify the values of $f$ and $g$ as follows to make them binary:
  If $x\in X$ is an isolated vertex of the graph, put $f(x) = 1$.
  Similarly, if $y\in Y$ is an isolated vertex of the graph, put $g(y) = 0$.
  For all other vertices, $r$-cyclic monotonicity guarantees that the function values are at most $1$ apart.
  This is because if there are two non-isolated vertices whose values are more than $1$ apart, then there must also be such vertices where $x\in X$, $y\in Y$ and $g(y) - f(x) > 1$ (if need be, we can pass from the relevant vertices to the other side using some edge in $U$).
  But then we can pick an incident edge for both $x$ and $y$, and this pair of edges violates $r$-cyclic monotonicity.
  As the function values are within $1$ of each other, we can add some constant to all function values of these vertices so that they are all $0$ or $1$.
  These modifications do not change the difference $\expect{\nu}{g} - \expect{\mu}{f}$, because the condition $\rho(U) = 1$ implies that all vertices in the supports of $\mu$ and $\nu$ are non-isolated, so that the same total gets added to both expected values and thus cancels out.
\end{proof}

\begin{lem}\label{lem:lift-two-valued-predicate}
  Let $X$ be a set, let $\mu\in\dfun X$ and let $f\colon X\to[0,1]$ be a map such that $f[X] = \{a,b\}$, where $a < b$.
  Then $\lambda(f)(\mu) = \min(b, \max(a, \mu(f^{-1}(b))))$.
\end{lem}
\begin{proof}
  First we note that for every $0 \le \epsilon \le 1$ we have:
  \begin{equation*}
    \mu(\{x \mid f(x) > \epsilon\}) =
    \begin{cases}
      1, & \text{if } 0 \le \epsilon < a \\
      \mu(f^{-1}(b)), & \text{if } a \le \epsilon < b \\
      0, & \text{if } b \le \epsilon \le 1
    \end{cases}
  \end{equation*}
  Therefore,
  \begin{align*}
    \lambda(f)(\mu) &= \inf_{0 \le \epsilon \le 1} \max(\epsilon, \mu(\{x \mid f(x) > \epsilon\})) \\
      &= \min(
        \inf_{0 \le \epsilon < a} \max(\epsilon, 1),
        \inf_{a \le \epsilon < b} \max(\epsilon, \mu(f^{-1}(b))),
        \inf_{b \le \epsilon \le 1} \max(\epsilon, 0)
      ) \\
      &= \min(1, \max(a,  \mu(f^{-1}(b))), b) \\
      &= \min(b, \max(a, \mu(f^{-1}(b)))). \qedhere
  \end{align*}
\end{proof}

\noindent
This allows us to establish duality:

\begin{thm}\label{thm:levy-prokhorov-duality-relational}
  For every $r\colon X\relto Y$ and every $\mu\in\dfun X$, $\nu\in\dfun Y$,
  \begin{equation*}
    \dlp(r)(\mu,\nu) = \kantrel{\lambda}(r)(\mu,\nu) = \wass{\lambda}(r)(\mu,\nu).
  \end{equation*}
\end{thm}
\begin{proof}
  Assume that $\dlp(r)(\mu,\nu) > 0$, and let $\epsilon > 0$ such that $\epsilon < \dlp(r)(\mu,\nu)$.
  We need to find an $r$-nonexpansive pair $(f,g)$ such that $\lambda(g)(\nu) - \lambda(f)(\mu) \ge \epsilon$.

  Using the optimal-transport representation of $\dlp(r)$ derived earlier, it follows that $\epsilon < \wass{\expectSymb}(r^\epsilon)(\mu,\nu)$.
  This is easiest understood using contraposition: if $\epsilon \ge \wass{\expectSymb}(r^\epsilon)(\mu,\nu)$, then $\epsilon \ge \dlp(r)(\mu,\nu)$.
  Hence there exists an $r^\epsilon$-nonexpansive pair $(p,q)$ such that $\expect{\nu}{q} - \expect{\mu}{p} \ge \epsilon$, where w.l.o.g.\ we may assume that $p$ and $q$ only take on the values $0$ and $1$ by~\cref{lem:01-price-functions}.

  Put $a = \expect{\mu}{p}$ and $b = \expect{\nu}{q}$.
  Now we define $f(x) = a$ whenever $p(x) = 0$ and $f(x) = a+\epsilon$ otherwise.
  Similarly we define $g(y) = a$ whenever $q(y) = 0$ and $g(y) = a+\epsilon$ otherwise.

  The pair $(f,g)$ is $r$-nonexpansive:
  The only relevant case is where $f(x) = a$ and $g(y) = a+\epsilon$, as in all other cases $g(y) - f(x) \le 0 \le r(x,y)$ trivially holds.
  In this case we have $p(x) = 0$ and $q(y) = 1$ and thus $1 = q(y) - p(x) \le r^\epsilon(x,y)$.
  By definition of $r^\epsilon$ this implies $g(y) - f(x) = \epsilon \le r(x,y)$, as required.

  It remains to show that $\lambda(g)(\nu) - \lambda(f)(\mu) \ge \epsilon$.
  To evaluate $f$ and $g$ under the predicate lifting, we use~\cref{lem:lift-two-valued-predicate}.
  Note that by definition of $f$ and $g$ we have that $\mu(f^{-1}(a+\epsilon)) = \mu(p^{-1}(1)) = \expect{\mu}{p} = a$ and similarly $\nu(g^{-1}(a+\epsilon)) = \expect{\nu}{q} = b$.
  Therefore, using the lemma, we obtain
  $\lambda(f)(\mu) = \min(a+\epsilon,\max(a,\mu(f^{-1}(a+\epsilon)))) = \min(a+\epsilon,\max(a,a)) = a$ and $\lambda(g)(\nu) = \min(a+\epsilon,\max(a,b)) = a+\epsilon$, concluding the proof.
\end{proof}

\noindent To prove that the duality carries over to the case of pseudometric liftings of set functors, we use a result~\cite{FuzzyLaxHemi} that relates $\kantrel{}$ to its pseudometric counterpart.
The most relevant consequence of this result for our purposes is the following:
\begin{lem}\label{lem:self-dual-rel-equals-sym}
  If $\lambda$ is a self-dual predicate lifting, then,
  for every pseudometric~$d$,
  \begin{equation*}
    \kantrel{\lambda}(d) = \kantsym{\lambda}(d).
  \end{equation*} 
\end{lem}
\begin{proof}
  This is an immediate consequence of~\cite[Lemma 5.10]{FuzzyLaxHemi}, because self-duality of $\lambda$ implies that the singleton set $\{\lambda\}$ is closed under duals.
\end{proof}

\begin{thm}\label{thm:levy-prokhorov-duality}
  For every pseudometric $d$ we have $\dlp(d) = \kantsym{\lambda}(d) = \wass{\lambda}(d)$.
\end{thm}
\begin{proof}
  We only need to show that $\kantrel{\lambda}(d) = \kantsym{\lambda}(d)$,
  which follows by~\cref{lem:lp-pred-lifting-self-dual,lem:self-dual-rel-equals-sym}.
\end{proof}

\newcommand*{\dfunpmet}{\dfun^\mathsf{Rad}}

\subsection{Radon measures}
For a pseudometric space $(X,d)$, we denote by $\Rad(X,d)$ the set of Radon probability measures, i.e.\ the probability measures $\mu$ on the $\sigma$-algebra generated by the open balls $\ball{d}{\epsilon}{x}$ (the \emph{Borel $\sigma$-algebra}) that are \emph{inner regular}, that is for every measurable set $A$,
\begin{equation*}
  \mu(A) = \sup \{\mu(K) \mid K\subseteq A, K\text{ measurable and compact} \}.
\end{equation*}
We write $\dfunpmet\colon\pmet\to\set$ for the functor assigning to a pseudometric space $(X,d)$ the set of Radon probability measures $\Rad(X,d)$.
We can lift this functor to a functor $\overline{\dfunpmet}\colon\pmet\to\pmet$ by equipping $\Rad(X,d)$ with Lévy-Prokhorov distance:
\begin{equation*}
  \dlp(d)(\mu,\nu) = \inf \{ \epsilon \mid \forall A\subseteq X \text{ measurable}.\; \mu(A) \le \nu(A^d_\epsilon) + \epsilon \}
\end{equation*}
We may also alternatively use the coupling-based Ky Fan distance, i.e.\ the construction $\wasspmet{\lambda}$, where the predicate lifting~$\lambda$ is once again defined as per~\eqref{eq:lp-pred-lifting-1}.
Note that this is indeed a valid definition and satisfies all the properties we need:
For $f$ nonexpansive, the sets $\{x \mid f(x) > \epsilon\}$ are all Borel-measurable, so we may apply any Radon measure to them.
This predicate lifting is also well-behaved; we can reuse the proof of \cref{lem:lp-pred-lifting-wb} by observing that all sets mentioned therein are Borel-measurable.
Finally, the functor $\dfunpmet$ preserves weak pullbacks and is product-agnostic, as discussed in \cref{expl:kant-wass-instances}.

There are duality results in the literature equating Lévy-Prokhorov distance and Ky Fan distance (e.g.~\cite{DudleyRealAnalysis}), but they are phrased for a slightly different setting, where instead of Radon measures one considers more general Borel measures (i.e.\ without the inner regularity condition) and in turn requires that the underlying (pseudo)metric space be \emph{separable} (admits a countable dense subset).
With some extra work, these can be carried over to our setting:
\begin{lem}\label{lem:dlp-is-wass}
  For any pseudometric space $(X,d)$ and Radon measures $\mu,\nu\in\Rad(X,d)$,
  \begin{equation*}
    \dlp(d)(\mu,\nu) = \wasspmet{\lambda}(d)(\mu,\nu).
  \end{equation*}
\end{lem}
\begin{proof}
  First, inner regularity of $\mu$ and $\nu$ allows us to construct a separable subspace $(Y,d|_Y)$ with $\mu(Y) = \nu(Y) = 1$.
  Explicitly, there exist compact sets $K_n$ ($n\in\nat$) such that $\mu(K_n) > 1 - \frac{1}{n}$ and $\nu(K_n) > 1 - \frac{1}{n}$, and we may put $Y$ as the closure of the union of the $K_n$.
  This is separable, as it is a countable union of compact sets, and $\mu(Y) = \nu(Y) = 1$ follows by monotonicity of $\mu$ and $\nu$.
  
  Neither the Lévy-Prokhorov distance nor the Ky Fan distance of $\mu$ and $\nu$ change when restricting to $Y$, because the set $X\setminus Y$ has measure $0$.
  We can therefore use the known result for the case of Borel measures on separable spaces~\cite{DudleyRealAnalysis}, which proves the equality.
\end{proof}

\noindent
Next, we show that the duality between Lévy-Prokhorov and Ky Fan distance implies generalized Kantorovich-Rubinstein duality for $\dfunpmet$ and $\lambda$.
To this end, we first show that predicates lifted by $\lambda$ are nonexpansive with respect to Lévy-Prokhorov distance
\begin{lem}\label{lem:lp-lifted-pred-nonexp}
  Let $(X,d)$ be a pseudometric space, and let $f\colon(X,d)\to([0,1],d_e)$ be nonexpansive.
  Then $\lambda_{(X,d)}(f)$ is nonexpansive as a map $(\dfunpmet(X,d),\dlp(d))\to([0,1],d_e)$.
\end{lem}
\begin{proof}
  Suppose that for $\delta > 0$ we have $\mu(A) \le \nu(A^d_\delta) + \delta$ for all Borel-measurable sets $A$.
  Then we also have $\nu(A) \le \mu(A^d_\delta) + \delta$ for all Borel-measurable sets $A$, by the same reasoning as in the proof of \cref{lem:lp-mirrored-clause-redundant}.
  We need to show that $|\lambda_{(X,d)}(f)(\nu) - \lambda_{(X,d)}(f)(\mu)| \le \delta$.
  We show that $\lambda_{(X,d)}(f)(\nu) - \lambda_{(X,d)}(f)(\mu) \le \delta$, the argument for the other inequality is symmetric.
  Expanding the definition of $\lambda$, we then need to show that
  \begin{equation*}
    \inf \{ \epsilon \mid \nu(f_\epsilon) \le \epsilon \}
    \le \inf \{ \epsilon \mid \mu(f_\epsilon) \le \epsilon \} + \delta,
  \end{equation*}
  where $f_\epsilon = \{x \in X \mid f(x) > \epsilon \}$.
  We do so by showing that every term on the right hand side also appears on the left hand side.
  Thus, let $\epsilon > 0$ be such that $\mu(f_\epsilon) \le \epsilon$; we need to show $\nu(f_{\epsilon+\delta}) \le \epsilon+\delta$.
  Because $f$ is nonexpansive, we have $(f_{\epsilon+\delta})^d_\delta \subseteq f_\epsilon$, and therefore
  \begin{equation*}
    \nu(f_{\epsilon+\delta})
    \le \mu((f_{\epsilon+\delta})^d_\delta) + \delta
    \le \mu(f_\epsilon) + \delta
    \le \epsilon + \delta. 
    \qedhere
  \end{equation*}
\end{proof}

\noindent We can now prove that $\dlp(d)$ is equal to the price-function-based lifting with respect to $\lambda$:
\begin{lem}\label{lem:dlp-is-kant}
  For any pseudometric space $(X,d)$ and Radon measures $\mu,\nu\in\Rad(X,d)$,
  \begin{equation*}
    \dlp(d)(\mu,\nu) = \kantsym{\lambda}(d)(\mu,\nu).
  \end{equation*}
\end{lem}
\begin{proof}
  There are two inequalities to show, where `$\ge$' follows immediately from \cref{lem:lp-lifted-pred-nonexp}.
  For `$\le$', let $\kantsym{\lambda}(d)(\mu,\nu) < \delta$.
  We need to show that for every Borel-measurable $A$, $\mu(A) \le \nu(A^d_\delta) + \delta$.
  So let $A$ be such a set, and assume w.l.o.g.~that $\mu(A) > \delta$.
  Define a nonexpansive map $f\colon(X,d)\to([0,1],d_e)$ by $f(x) = \mu(A) \ominus d(x,A)$, where $d(x,A) = \inf \{d(x,y) \mid y\in A\}$.
  Then
  \begin{equation*}
    \lambda_{(X,d)}(f)(\mu)
    = \inf \{ \epsilon \mid \mu(f_\epsilon) \le \epsilon \}
    = \inf \{ \epsilon \mid \mu(A) \le \epsilon \} = \mu(A),
  \end{equation*}
  where in the second step we use that $f_\epsilon=\emptyset$ for $\epsilon \ge \mu(A)$ and $f_\epsilon \supseteq A$ otherwise.
  From the assumption that $\kantsym{\lambda}(d)(\mu,\nu) < \delta$ we have
  $|\lambda_{(X,d)}(f)(\nu) - \lambda_{(X,d)}(f)(\mu)| < \delta$ and therefore
  \begin{equation*}
    \inf \{ \epsilon \mid \nu(f_\epsilon) \le \epsilon \} = \lambda_{(X,d)}(f)(\nu) > \lambda_{(X,d)}(f)(\mu) - \delta = \mu(A) - \delta.
  \end{equation*}
  Because of the strict inequality, we know that $\mu(A) - \delta$ is not a member of the set in the leftmost term, and therefore $\nu(f_{\mu(A)-\delta}) > \mu(A) - \delta$.
  Next, we have $f_{\mu(A)-\delta} \subseteq A^d_\delta$, because
  \begin{equation*}
    x \in f_{\mu(A)-\delta}
    \iff f(x) > \mu(A)-\delta
    \iff d(x,A) < \delta
    \implies x\in A^d_\delta.
  \end{equation*}
  Therefore,
  \begin{math}
    \nu(A^d_\delta) \ge \nu(f_{\mu(A)-\delta}) > \mu(A) - \delta,
  \end{math}
  and rearranging terms finishes the proof.
\end{proof}

\noindent Combining \cref{lem:dlp-is-wass,lem:dlp-is-kant} together:
\begin{thm}\label{thm:lp-duality-radon}
  For any pseudometric space $(X,d)$ and Radon measures $\mu,\nu\in\Rad(X,d)$,
  \begin{equation*}
    \dlp(d)(\mu,\nu) = \kantsym{\lambda}(d)(\mu,\nu) = \wasspmet{\lambda}(d)(\mu,\nu).
  \end{equation*}
\end{thm}

\section{Convex Powerset Functor}
\label{sec:convex-powerset}

We will next tackle duality for the case of the convex powerset
functor, a functor that has been studied in depth for modelling
systems combining probability and non-determinism
(e.g.,~\cite{Bonchi_et_al:PowerConvexAlg}).  A non-empty set
$D\subseteq \dfun X$ of probability distributions is \emph{convex} if
for all $\mu_1,\mu_2\in D$ it also holds that
$\mu_1 +_p \mu_2 := p\cdot \mu_1 + (1-p)\cdot \mu_2\in D$ (where
$p\in[0,1]$).  For a set~$X$, we define
\begin{equation*}
  \cfun X = \{\emptyset \neq D\subseteq \dfun X \mid \text{$D$ is
      convex} \}.
\end{equation*}
Of course, we have $\cfun X \subseteq \pfun\dfun X$ for every set $X$.
In fact, it is easily verified that every map $\pfun\dfun f$ preserves convex sets, so that we obtain a subfunctor~$\cfun$ of the composite functor $\pfun\dfun$.

A straightforward -- but futile -- approach to prove duality for
$\cfun$ would be to observe that it holds for the powerset and
distribution functor, and then apply a compositionality
result. However, the two notions of lifting (price-function-based and coupling-based) are quite fragile when it comes to compositionality, i.e.\ it
does not hold in general that the composite of liftings of functors
$F$, $G$ based on the predicate liftings $\lambda^F$, $\lambda^G$ is
the lifting of the composite $FG$ (based on the obvious combined
modality $\lambda^{FG}_X = \lambda^F_{GX}\circ \lambda^G_X$)
\cite{bbkk:coalgebraic-behavioral-metrics}. While it is known that the
coupling-based lifting of the convex powerset functor arises by combining
the coupling-based liftings of the component functors
\cite{FuzzyLaxHemi}, this is incorrect for the price-function-based
lifting~\cite{dgkknrw:behavioural-metrics-compositionality}. In fact,
the given counterexample uses a set that is \emph{not} convex, thus
suggesting that the problem might disappear if we restrict to convex
sets.

We use the $\sup$ modality for the powerset functor and the
expectation ($\expectSymb$) modality for the distribution functor. Our
aim is to establish the rather non-trivial result that the combined
modality $\lambda_X(f)(A) = \sup \{ \expect{\mu}{f} \mid \mu\in A \}$
is indeed expressive for the convex powerset functor $\cfun$ on its
own.

As before, $\kantsym{\lambda} \le \wass{\lambda}$ holds in categorical
generality, so the main task is to prove the converse inequality
$\wass{\lambda} \le \kantsym{\lambda}$.  As mentioned, it is known
that the coupling-based representation in terms of $\lambda$
decomposes into the coupling-based representations in terms of $\sup$
and $\expectSymb$, i.e.\ the Hausdorff and Kantorovich-Rubinstein
liftings, respectively:
\begin{equation*}
  \wass{\lambda}(d) = \wass{\sup}(\wass{\expectSymb}(d)) = \dh(\dk(d)) =: \dhk(d)
\end{equation*}

% \subsection{Duality}

\noindent
To obtain the duality result, it will be convenient to pass from
pseudometric spaces to metric spaces.  Recall that the \emph{metric
  quotient} of a pseudometric space $(X,d)$ is the metric space
$(X_\sim,d_\sim)$ where $X_\sim$ is the set of equivalence classes of
the equivalence relation $x \sim y \iff d(x,y) = 0$, and
$d_\sim([x],[y]) = d(x,y)$ for any two equivalence classes
$[x],[y]\in X_\sim$. The following is immediate from the fact that
both price-function-based and coupling-based liftings preserve
isometries~\cite{bbkk:coalgebraic-behavioral-metrics}:

\begin{lem}\label{lem:convex-metric-quotient}
  Let $(X,d)$ be a pseudometric space, let $(X_\sim,d_\sim)$ be its metric quotient, and let $\pi\colon X\to X_\sim, x\mapsto [x]$.
  We then have, for every $A,B\in\cfun X$,
  \begin{equation*}
    \kantsym{\lambda}(d)(A,B) = \kantsym{\lambda}(d_\sim)(A_\sim,B_\sim)
    \quad\text{and}\quad
    \wass{\lambda}(d)(A,B) = \wass{\lambda}(d_\sim)(A_\sim,B_\sim),
  \end{equation*}
  where $A_\sim = \cfun\pi(A)$ and $B_\sim = \cfun\pi(B)$.  
\end{lem}

\noindent Using the above lemma, we may therefore from now on assume that we are working over a metric space $(X,d)$.
We may also assume that $X\neq\emptyset$, as otherwise $\cfun X=\emptyset$ and both coupling-based and price-function-based distance are the empty metric, hence equal.

The main intuition behind the proof is best understood in the case where $X = \{x_1,\dots,x_n\}$ is finite, even though the proof will work for arbitrary $X$.
In this case we may view probability distributions and fuzzy predicates on $X$ as vectors in $\real^n$, and the expectation modality simply computes the dot product between two such vectors:
\begin{math}
  \expect{\mu}{f} = \mu(x_1)\cdot f(x_1) + \dots + \mu(x_n)\cdot f(x_n).
\end{math}
If $A,B\in \cfun X$ satisfy $\dhk(d)(A,B) > \epsilon$, then this means, by the definition of the Hausdorff distance, that there must be some $\mu\in A$ such that $\dk(\mu,\nu) > \epsilon$ for all $\nu\in B$ (or that we are in the symmetric situation with $A$ and $B$ swapped).
This implies that $B$ and the $\epsilon$-ball around $\mu$ are disjoint convex sets, so we can apply the hyperplane separation theorem to find a hyperplane $H$ such that the two sets lie on opposite sides of that plane (more precisely, $B$ may intersect the hyperplane, but the $\epsilon$-ball may not).
A price function witnessing distance at least $\epsilon$ under the combined modality~$\lambda$ can now be constructed from the normal vector of~$H$.

Our proof mostly follows the outline above, but because we may now be working with infinite-dimensional vector spaces, some functional analysis will be required.
We leverage this added complexity in~\cref{sec:borel}, where we show that the duality result remains true when passing from discrete probability measures to Radon probability measures.
In fact, for many of the intermediate facts and lemmas that feed into the proof of \cref{thm:convex-powerset-duality} their proof will simply refer to the proof of the respective statement about Radon probability measures in~\cref{sec:borel}.

We fix a point $x_0\in X$ and consider the vector space $\Lip_0(X)$ consisting of the real-valued Lipschitz functions on $X$ vanishing at $x_0$:
\begin{equation*}
  \Lip_0(X) = \{ f\colon X\to\real \mid f(x_0) = 0, \sup_{x\neq y} \textstyle\frac{f(y)-f(x)}{d(x,y)} < \infty \}.
\end{equation*}
This is a Banach space with norm given by $\lipnorm{f} = \sup_{x\neq y} \frac{f(y)-f(x)}{d(x,y)}$.
We will construct our price function in this space, which is made possible by the fact that the set of probability distributions can be mapped into its dual:
\begin{lem}\label{lem:embed-into-dual-discrete}
  The set $\dfun X$ embeds into the continuous dual space $\Lip_0(X)^*$, that is, every discrete probability measure $\mu$ gives rise to a continuous linear functional $L_\mu\colon \Lip_0(X)\to\real$, which may explicitly be given by $L_\mu(f) = \expect{\mu}{f}$.
\end{lem}
\begin{proof}
  This follows by \cref{lem:embed-into-dual}, because every discrete probability measure is also a Radon measure.
\end{proof}

\noindent
In what follows, we often do not distinguish between $\mu$ and $L_\mu$ and treat $\dfun X$ as a subset of $\Lip_0(X)^*$.
We equip $\Lip_0(X)^*$ with the \emph{weak-* topology}, which is the weakest topology on $\Lip_0(X)^*$ making all the maps $\psi \mapsto \psi(f)$ for $f\in\Lip_0(X)$ continuous.
Equivalently, this is the \emph{initial topology} wrt.\ the maps $\psi \mapsto \psi(f)$.
Crucially, this topology coincides with the one given by the Kantorovich-Rubinstein distance:

\begin{lem}\label{lem:topologies-coincide-discrete}
  Let $(\mu_n)_{n\in\nat}$ be a sequence in $\dfun X$ and let $\mu\in\dfun X$.
  Then $\mu_n\to\mu$ in the topology given by $\dk(d)$ iff $L_{\mu_n}\to L_\mu$ in the weak-* topology.
\end{lem}
\begin{proof}
  This follows by Lemma~\ref{lem:topologies-coincide}, again because every discrete probability measure is a Radon measure.
\end{proof}

\noindent
The space $\Lip_0(X)^*$ is normed via the \emph{operator norm} $\opnorm{\psi} = \sup_{\lipnorm{f} \le 1} \psi(f)$.
This norm relates to the Kantorovich-Rubinstein metric as follows:

\begin{lem}\label{lem:wasserstein-norm-discrete}
  For any $\mu,\nu\in\dfun X$ we have $\dk(d)(\mu,\nu) = \opnorm{\nu-\mu}$.
\end{lem}
\begin{proof}
  This follows by \cref{lem:wasserstein-norm}, again because every discrete probability measure is a Radon measure.
\end{proof}

\noindent
We need one more fact about $\epsilon$-balls with respect to Kantorovich-Rubinstein distance, which  characterizes them as the convex hulls of specific vertices.
\begin{lem}\label{lem:kantorovich-ball-discrete}
  Let $\epsilon > 0$.
  For $x,y\in X$, $x\neq y$, let
  \begin{math}
    \textstyle\mu_{xy} = \frac{\epsilon}{d(x,y)} (\delta_y - \delta_x),
  \end{math}
  where $\delta_x$ and $\delta_y$ are Dirac measures, and put $C = \convx \{ \mu_{xy} \mid x,y\in X, x \neq y\}$.
  Then we have, for any $\mu,\nu\in\dfun X$, that
  $\dk(d)(\mu,\nu) \le \epsilon \iff \nu-\mu \in \overline{C}$.
\end{lem}
\begin{proof}
  This follows by \cref{lem:kantorovich-ball}, as discrete probability measures are Radon measures.
\end{proof}

\noindent
We record the following fact about the weak-* topology:

\begin{lem}\label{lem:weak-star-representation}
  Let $V$ be a topological vector space, and let $V^*$ be its continuous dual space, equipped with the weak-* topology.
  Then for every continuous linear functional $\psi\colon V^* \to\real$ there exists $v\in V$ such that $\psi(f) = f(v)$ for every $f\in V^*$.
\end{lem}
\begin{proof}
  Let $\psi\colon V^*\to\real$ be a continuous linear functional.
  Then the set $\{f \in V^* \mid |\psi(f)| < 1 \}$ is open and therefore contains a weak-* neighbourhood of $0$.
  This means that there must be $v_1,\dots,v_n\in V$ such that
  \begin{equation}\label{eq:nbhood-equal}
    \bigcap_{1 \le k \le n} \{f\in V^* \mid |f(v_k)| < 1\} \subseteq
    \{f \in V^* \mid |\psi(f)| < 1 \}.
  \end{equation}
  We can now show that $|\psi(f)| \le \max_{1 \le k \le n} |f(v_k)|$:
  Assume that $m > |f(v_k)|$ for all $k$.
  Then $\frac{1}{m}\cdot f$ is a member of the set on the left of~\eqref{eq:nbhood-equal} and hence a member of the set on the right.
  Therefore $|\psi(f)| < m$ by homogeneity.

  This implies that for every $f\in V^*$, if $f(v_k) = 0$ for all $k$, then $\psi(f) = 0$ as well, or, put differently, $\bigcap_{1 \le k \le n} \ker v_k \subseteq \ker\psi$.
  It follows~\cite[Theorem 3.20]{HoffmannKunze71} that $\psi$ is a linear combination of $v_1,\dots,v_n$.
\end{proof}

% \begin{thm}[Banach-Alaoglu] \label{thm:banach-alaoglu}
%   Let $V$ be a topological vector space, and let $V^*$ be its continuous dual space, equipped with the weak-* topology.
%   Then the unit ball $\{ f\in V^* \mid \opnorm{f} \le 1 \}$ is compact.
% \end{thm}

% \pwnote[inline]{Find a fuzzy-relational version.}
\noindent
We are now in a position to state and prove the duality result:

\begin{thm}\label{thm:convex-powerset-duality}
  For any two convex sets $A,B \in\cfun X$,
  \begin{equation*}
    \dhk(d)(A,B) = \kantsym{\lambda}(d)(A,B).
  \end{equation*}
\end{thm}
% \begin{proofsketch}
%   The inequality $\kantsym{\lambda} \le \dhk$ follows from previous
%   results \cite{bbkk:coalgebraic-behavioral-metrics}, hence it is
%   sufficient to show $\dhk(d)(A,B) \le \kantsym{\lambda}(d)(A,B)$.
  
%   Let $0 < \epsilon < \dhk(d)(A,B)$.
%   As outlined above, we may  assume wlog.~that there exists some $\mu\in A$ such that $\dk(d)(\mu,\nu) > \epsilon$ for every $\nu\in B$.
%   Let $C$ be the closed $\epsilon$-ball around $\mu$, shifted by $-\mu$, and let $D$ be the closure of $B$, also shifted by $-\mu$.
%   These two sets are closed and convex, so by the Hahn-Banach separation theorem there exists a continuous linear functional $g\colon\Lip_0(X)^* \to\real$ and $c\in\real$ such that
%   \begin{equation*}
%     \sup_{\nu\in C} g(\nu) < c \le \inf_{\nu\in D} g(\nu),
%   \end{equation*}
%   and because we are in the weak-* topology, the functional $g$ can be represented in the form $\mu\mapsto\expect{\mu}{f}$ for some $f\in\Lip_0(X)$.
%   We replace $f$ by $f_1 = -f/\lipnorm{f}$, which results in a nonexpansive function for which
%   \begin{equation*}
%     |\sup_{\nu\in B} \expect{\nu}{f_1} - \sup_{\nu\in A} \expect{\nu}{f_1}| \ge \epsilon.
%   \end{equation*}
%   The range of $f_1$ is not necessarily contained in $[0,1]$, but it must be contained in some subinterval of $\real$ of length at most $1$ by nonexpansiveness and because $d$ is $1$-bounded.
%   As expectation is linear, we may simply shift $f_1$ by a suitable amount to extract the desired price function.
% \end{proofsketch}
\begin{proof}
  The inequality $\kantsym{\lambda} \le \dhk$ follows from previous
  results \cite{bbkk:coalgebraic-behavioral-metrics}, hence it is
  sufficient to show $\dhk(d)(A,B) \le \kantsym{\lambda}(d)(A,B)$.
  
  If $\dhk(d)(A,B) = 0$, then there is nothing to show.
  Otherwise, let $0 < \epsilon < \dhk(d)(A,B)$, and assume w.l.o.g.\ that
  \begin{equation*}
    \dhk(d)(A,B) = \adjustlimits\sup_{\mu\in A}\inf_{\nu\in B} \dk(d)(\mu,\nu).
  \end{equation*}
  By assumption there exists some $\mu\in A$ such that $\dk(d)(\mu,\nu) > \epsilon$ for every $\nu\in B$.
  Let $C$ be the closed $\epsilon$-ball around $\mu$, shifted by $-\mu$, that is,
  $C = \{ \nu-\mu \mid \nu \in \dfun(X), \dk(d)(\mu,\nu) \le \epsilon \}$.
  The set $C$ is convex and closed, and moreover it is compact by the Banach-Alaoglu theorem. %\pwnote{Expand this}
  Let $D$ be the closure of $B$, also shifted by $-\mu$, that is, $D = \overline{B} - \mu$.
  As the closure of every convex set is convex, $D$ is convex as well.
  By the Hahn-Banach separation theorem there exists a continuous linear functional $g\colon\Lip_0(X)^* \to\real$ and $c\in\real$ such that
  \begin{equation*}
    \sup_{\nu\in C} g(\nu) < c \le \inf_{\nu\in D} g(\nu).
  \end{equation*}
  Shifting back by $\mu$ and passing from the closure of $B$ back to $B$, we obtain
  \begin{equation*}
    \sup_{\nu\colon \dk(d)(\mu,\nu) \le \epsilon} g(\nu) < g(\mu) + c \le \inf_{\nu\in B} g(\nu).
  \end{equation*}
  Because we are in the weak-* topology, the functional $g$ can be represented in the form $\mu\mapsto\expect{\mu}{f}$ for some $f\in\Lip_0(X)$ (Lemma~\ref{lem:weak-star-representation}), so
  \begin{equation*}
    \sup_{\nu\colon \dk(d)(\mu,\nu) \le \epsilon} \expect{\nu}{f}-\expect{\mu}{f} < c \le \inf_{\nu\in B} \expect{\nu}{f} - \expect{\mu}{f}.
  \end{equation*}
  We replace $f$ by $f_1 = f/\lipnorm{f}$, which has norm $1$ and is hence nonexpansive by construction:
  \begin{equation}\label{eqn:hahn-banach}
    \sup_{\nu\colon \dk(d)(\mu,\nu) \le \epsilon} \expect{\nu}{f_1}-\expect{\mu}{f_1} < c/\lipnorm{f} \le \inf_{\nu\in B} \expect{\nu}{f_1} - \expect{\mu}{f_1}.
  \end{equation}
  For the leftmost term we have, by \cref{lem:kantorovich-ball} (and reusing the notation $\mu_{xy}$ from there),
  \begin{equation*}
    \sup_{\nu\colon \dk(d)(\mu,\nu) \le \epsilon} \expect{\nu}{f_1}-\expect{\mu}{f_1}
    = \sup_{x\neq y} \expect{\mu_{xy}}{f}
    = \sup_{x\neq y} \textstyle\frac{\epsilon}{d(x,y)}(f(y)-f(x))
    = \epsilon \lipnorm{f} = \epsilon.
  \end{equation*}
  Combining this with \eqref{eqn:hahn-banach}, we have $\epsilon + \expect{\mu}{f_1} \le \inf_{\nu\in B} \expect{\nu}{f_1}$.
  We also have $\inf_{\nu\in A} \expect{\nu}{f_1} \le \expect{\mu}{f_1}$ because $\mu\in A$.
  Finally, we replace $f_1$ by the function $f'_1(x) = \sup_{x'\in X} f_1(x') - f_1(x)$.
  Then $f'_1$ is also nonexpansive and its range is therefore in $[0,1]$, because $d$ is a $1$-bounded metric.
  \begin{multline*}
    \kantsym{\lambda}(d)(A,B)
    \ge |\sup_{\nu\in B} \expect{\nu}{f'_1} - \sup_{\nu\in A} \expect{\nu}{f'_1}| \\
    = |\inf_{\nu\in B} \expect{\nu}{f_1} - \inf_{\nu\in A} \expect{\nu}{f_1}|
    \ge \expect{\mu}{f_1} + \epsilon - \expect{\mu}{f_1} = \epsilon.
  \end{multline*}
  From the first to second line we used linearity of expectation, which causes the constant term $\sup_{x'\in X} f_1(x')$ to cancel and the suprema to flip to infima.
\end{proof}

\begin{rem}
  A natural question to ask is whether one can, like in the previous section, also obtain a fuzzy-relational version of the duality result.
  We expect the answer to be `yes' in the sense of~\eqref{eq:duality-rel}, i.e.\ when additionally considering the dual predicate lifting $\kappa_X(f)(A) = \inf \{\expect{\mu}{f} \mid \mu\in A\}$ on the price-function-based side.
  We leave this question open for now.
\end{rem}

\begin{rem}[Compositionality]
  It has been shown in work on \emph{correspondences} between
  price-function-based and coupling-based representations of metric
  liftings~\cite{HumeauEA25} (cf.~\autoref{sec:introduction} under
  \emph{Related Work}) that such correspondences can be combined along
  sums and products of functors, so that one arrives at general
  correspondence results for classes of functors obtained by closing
  given basic building blocks (originally constant functors, identity,
  powerset, and distributions) under sum and product. One thus obtains
  correspondences for composite system types such as labelled Markov
  chains~\cite[Example~41]{HumeauEA25}. The correspondences produced
  in this way are not generalized Kantorovich-Rubinstein dualities in
  the strict sense we use here, as the transition from the coupling-based
  presentation to the price-function-based presentation in general involves the
  introduction of additional modalities. In particular, this happens
  for products, where one needs to introduce separate modalities for
  the factors (indeed, this is what is morally behind the fact that
  generalized Kantorovich-Rubinstein duality fails for the squaring
  functor~\cite{bbkk:coalgebraic-behavioral-metrics}). Nevertheless,
  our results on generalized Kantorovich-Rubinstein duality for
  L\'evy-Prokhorov distance and convex powerset imply that these
  functors can now be used as additional basic building blocks in this
  framework.
\end{rem}

\subsection{Algorithmic considerations}

A nice aspect of the duality result for the convex powerset functor is that it can be used as the basis of an algorithm to compute values of the distance $\dhk(d)$.
Explicitly, if $(X,d)$ is a \emph{finite} pseudometric space, and $A_0,B_0\subseteq\dfun X$ are \emph{finite} sets of probability measures, the problem is to compute the distance $\dhk(d)(A,B)$, where $A = \convx(A_0)$ and $B = \convx(B_0)$.
The distance expands as follows:
\begin{equation*}
  \dhk(d)(A,B) = \max(\adjustlimits\sup_{\mu\in A} \inf_{\nu\in B} \dk(d)(\mu,\nu),\adjustlimits\sup_{\nu\in B}\inf_{\mu\in A} \dk(d)(\mu,\nu)).
\end{equation*}
As the map $\mu \mapsto \inf_{\nu\in B} \dk(d)(\mu,\nu)$ is convex, the left supremum can instead be taken over $A_0$ without changing the value, and similarly we may take the right supremum over $B_0$.
It is however not in general true that the infima may be taken over $B_0$ and $A_0$, respectively.
\begin{expl}\label{expl:convex-powerset}
  Let $X = \{x,y,z\}$, and assume $d(x,y) = d(x,z) = d(y,z) = 1$.
  Let $A_0 = \{\mu_0,\mu_1\}$,
  $B_0 = \{\mu_2,\mu_3\}$,
  $\mu_0 = \frac{1}{3}\cdot x + \frac{1}{3}\cdot y + \frac{1}{3}\cdot z$,
  $\mu_1 = \frac{2}{3}\cdot x + \frac{1}{3}\cdot y$,
  $\mu_2 = \frac{2}{3}\cdot y + \frac{1}{3}\cdot z$,
  $\mu_3 = \frac{2}{3}\cdot z + \frac{1}{3}\cdot x$.
  Then the minimal distance from $\mu_1$ to $B$ is witnessed by
  $\mu^* = \frac{1}{2}\cdot \mu_2 + \frac{1}{2}\cdot \mu_3 =
  \frac{1}{6}\cdot x + \frac{1}{3}\cdot y + \frac{1}{2}\cdot z$.
  See~\Cref{fig:convex-powerset} for an illustration.

  \begin{figure}[ht]
    \centering
    \includegraphics[width=0.5\textwidth]{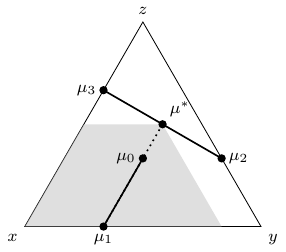}
    \caption{Illustration of \cref{expl:convex-powerset}. The thick line segments correspond to the sets $A = \convx(A_0)$ and $B = \convx(B_0)$, while the shaded region shows the radius-$\frac{1}{2}$-ball around~$\mu_1$, which is part of a regular hexagon centered at $\mu_1$.}
    \label{fig:convex-powerset}
  \end{figure}
\end{expl}
% \pwnote{@Paul: expand this example to explain why balls are hexagons, add another figure based on the ones found in \texttt{experiments/visualize-kantorovich}}

\noindent
To compute the inner infimum, therefore, a more sophisticated approach is required.
\c{C}elik et al.~\cite{CelikEA20} consider the problem of solving an optimal transport problem between probability distributions $\mu$ and $\nu$, where $\mu$ is fixed and $\nu$ ranges over some algebraic variety.
Their methods can be adapted to our setting, where~$\nu$ instead ranges over a convex set~$B$:

Consider the typical flow network that one constructs to solve the optimal transport problem, i.e.\ a complete bipartite graph whose partitions are two copies of the set $X$, which we think of as representing the two probability distributions $\mu$ and $\nu$.
A transport plan consists in assigning weights to the edges of this graph in such a way that the sum of weights of incident edges for each vertex matches its probability.
The key idea is that there always exists an optimal transport plan that is acyclic in the sense that the edges with positive weight do not contain any cycle.
Moreover, if the tree of edges used by the transport plan is known, then the weights of the plan are uniquely determined by the probabilities in $\mu$ and $\nu$ and can be computed by a depth-first (or breadth-first) traversal of the tree.
This means that the distance between $\mu$ and $\nu$ can be computed by enumerating all spanning trees of the complete bipartite graph, computing the weights, and taking the least value over all trees where the weights are all non-negative.
If $\mu\in A_0$ is fixed and $\nu$ ranges over $B$, then the weights in the coupling corresponding to a given spanning tree are linear combinations of the probabilities in $\nu$.
The conditions that these weights are non-negative define a linear program whose variables are the coefficients in the convex combination of the elements of~$B_0$ and whose constraints state that the weights in the tree are all non-negative, and the distance can be found by solving all these linear programs.

The algorithm we just described requires solving exponentially many linear programs in the size of $X$, as there is one such linear program for each spanning tree.
If we expand the dual representation of the distance between $A$ and $B$ instead, we obtain
\begin{equation*}
  \kantsym{\lambda}(d)(A,B) =
  \sup \{ |\sup_{\nu\in B} \expect{\nu}{f} - \sup_{\mu\in A} \expect{\mu}{f}|
    \mid f\colon(X,d)\to(\real,\de) \text{ nonexp.}\}.
\end{equation*}
This quantity is much easier to compute, as there are no nested suprema/infima.
Similarly as before, since expectation is linear, the two suprema above may instead be taken over $A_0$ and $B_0$, respectively.
To compute the distance, we can employ the following algorithm.
Loop over all pairs $(\mu_0,\nu_0)\in A_0\times B_0$.
Given $\mu_0$ and $\nu_0$, the subproblem is then to find the supremum above for all the nonexpansive $f$ such that $\expect{\mu_0}{f} = \sup_{\mu\in A}\expect{\mu}{f}$ and $\expect{\nu_0}{f} = \sup_{\nu\in B}\expect{\nu}{f}$.
This subproblem can be rephrased as a linear program over the variables $(f_x)_{x\in X}$ corresponding to the function values of $f$.
Nonexpansiveness of $f$, the constraints $\expect{\mu_0}{f} \ge \expect{\mu}{f}$ for $\mu\in A_0$ and the similar constraints for the $\nu\in B_0$ are all easily expressed as linear inequalities between the $f_x$.
The objective function is $|\expect{\nu_0}{f} - \expect{\mu_0}{f}|$, which is the maximum of two linear expressions over the $f_x$, so we can simply solve the linear program twice, once for each of the two expressions.
This algorithm has a runtime complexity that is polynomial in $|X|$, $|A_0|$ and $|B_0|$, a clear improvement over the exponential complexity for the previous approach.

\subsection{Radon Measures}
\label{sec:borel}

\newcommand*{\cfunpmet}{\cfun^\mathsf{Rad}}

We now consider convex powerset duality for the more general case of continuous probability measures, in particular Radon measures.
% We can equip this set with the Lévy-Prokhorov distance, which is the same definition as before, except that we restrict to measurable sets:
% \begin{equation*}
%   \dlp(d)(\mu,\nu) = \inf \{ \epsilon > 0 \mid \forall A\subseteq X \text{ measurable}.\; \mu(A) \le \nu(A^d_\epsilon) + \epsilon \}.
% \end{equation*}
% This results in an endofunctor $\dfunpmet\colon\pmet\to\pmet$ given by $\dfunpmet(X,d) = (\Bor(X,d),\dlp(d))$.
Every convex combination of Radon measures is itself a Radon measure.
We can therefore also define a functor $\cfunpmet\colon\pmet\to\set$ where $\cfunpmet(X,d)$ is the set of non-empty convex subsets of $\Rad(X,d)$.
We emphasize that $\cfunpmet$ is not the lifting of a set functor, as the definition of the underlying set depends on the pseudometric on the base space.

We can lift this functor to a functor $\overline{\cfunpmet}\colon\pmet\to\pmet$ by equipping $\cfunpmet (X,d)$  with the Hausdorff-Kantorovich distance $\dhk(d) = \dh(\dk(d))$, where both $\dh$ and $\dk$ are defined as before.
Note that every nonexpansive map on $(X,d)$ is Borel-measurable, so that no issues arise when taking expected values. We restate the definitions explicitly:
\begin{gather*}
  \dk(d)(\mu,\nu) = \sup \{ |\expect{\nu}{f} - \expect{\mu}{f}| \mid f\colon (X,d)\to([0,1],\de) \text{ nonexpansive}\} \\
  \dh(d)(A,B) = \max(\adjustlimits \sup_{x\in A}\inf_{y\in B} d(x,y), \adjustlimits\sup_{y\in B}\inf_{x\in A} d(x,y)).
\end{gather*}

\noindent
By general results~\cite{KantorovichFunctors}, we know that $\cfunpmet$ admits a Kantorovich representation for some suitable class of predicate liftings.
% A predicate lifting for a functor $F\colon\pmet\to\pmet$ is a natural transformation $\lambda\colon\pmet(-,\real^\kappa)\to\pmet(F-,\real)$,\pwnote{I'm not actually sure what number space to put there\dots} where $\kappa$ is some cardinal number.
Specifically, it suffices by~\cite[Theorem 5.3]{KantorovichFunctors}
to show that $\cfunpmet$ preserves \emph{initial morphisms}.  In the
real-valued setting these correspond to \emph{isometries}, i.e.\
nonexpansive maps $f\colon(X,d_1)\to(Y,d_2)$ such that
$d_2(f(x),f(x')) = d_1(x,x')$ for all~$x,x'\in X$. Indeed we have
\begin{lem}\label{lem:isometries}
  $\cfunpmet$ preserves isometries.
\end{lem}
\begin{proof}
  We know that $\dhk(d)$ arises as the composition of two Kantorovich constructions for the predicate liftings $\sup$ and $\expectSymb$, which, by the mentioned~\cite[Theorem 5.3]{KantorovichFunctors} individually preserve isometries, so their composition does as well.
\end{proof}

\noindent
The class of predicate liftings one obtains in this way is quite large; we show
that one can in fact make do with just a single predicate lifting,
which strengthens the corresponding instance of the coalgebraic
quantitative Hennessy-Milner
theorem~\cite[Corollary~5.10]{KantorovichFunctors} by providing a compact
explicit syntax for the expressive logic.
As in the discrete version of the proof, we write $\lambda$ for the predicate lifting for $\cfunpmet$ that is given as the composite of the supremum and expectation modalities:
\begin{equation*}
  \lambda_{(X,d)}(f)(A) = \sup_{\mu\in A} \expect{\mu}{f}.
\end{equation*}

\noindent
The functor $\cfunpmet$ is product-agnostic, as, like in the probabilistic setting, the set of Radon measures on $(X,d_1)\times(Y,d_2)$ is exactly equal to the set of Radon measures on $(X,d_1)\boxplus(Y,d_2)$, so the respective sets of convex subsets are of course equal as well.
For the second part of the property, we once more have that the isomorphism $\cfunpmet\langle\pi_1,\pi_2\rangle$ is in fact an identity map, which easily implies the claim.

Next, we retrace some of the proof steps from the discrete case, with the goal of establishing that $\dhk(d) = \kantsym{\lambda}(d)$ for every pseudometric $d$.
First, we note that we can, again, pass to the metric quotient, as \cref{lem:isometries} immediately implies the following generalization of \cref{lem:convex-metric-quotient}:
\begin{lem}
  Let $(X,d)$ be a pseudometric space, and let $(X_\sim,d_\sim)$ be its metric quotient, and let $\pi\colon X\to X_\sim$, $x\mapsto [x]$.
  We then have, for every $A,B\in\cfunpmet X$,
  \begin{equation*}
    \kantsym{\lambda}(d)(A,B) = \kantsym{\lambda}(d_\sim)(A_\sim,B_\sim)
    \quad\text{and}\quad
    \dhk(d)(A,B) = \dhk(d_\sim)(A_\sim,B_\sim),
  \end{equation*}
  where $A_\sim = \cfunpmet\pi(A)$ and $B_\sim = \cfunpmet\pi(B)$.
\end{lem}
% \begin{proof}
%   The proof from \cref{lem:convex-metric-quotient} (for $\kantsym{\lambda}$) can be reused without change.
%   For the construction $\dhk$ we rely on the other proof from the same lemma (for $\wass{\lambda}$), but it is a bit less obvious that there is a bijection between the sets of nonexpansive maps.
%   It is still true, however, that for all $\mu,\mu'\in\Rad(X,d)$ such that $\pi(\mu) = \pi(\mu')$ we have $\dk(d)(\mu,\mu') = 0$ and hence $f(\mu) = f(\mu')$ for every nonexpansive $f$.
% \end{proof}
\noindent Further key facts generalize as follows:
\begin{lem}\label{lem:embed-into-dual}
  The set $\Rad(X,d)$ embeds into the continuous dual space $\Lip_0(X)^*$, that is, every Radon probability measure $\mu$ gives rise to a continuous linear functional $L_\mu\colon \Lip_0(X)\to\real$, which may explicitly be given by $L_\mu(f) = \expect{\mu}{f}$.
\end{lem}
\begin{proof}
  Linearity of $L_\mu$ is clear.
  For continuity, note that over a normed space (such as a Banach space) a linear functional is continuous iff it is bounded, so we show boundedness.
  We have
  \begin{equation*}
    |\expect{\mu}{f}|
    \le \sint{X}{|f(x)|}{\mu(x)}
    \le \sint{X}{\lipnorm{f} d(x,x_0)}{\mu(x)}
    =  \lipnorm{f} \sint{X}{d(x,x_0)}{\mu(x)},
  \end{equation*}
  so $L_\mu$ is bounded because $\mu$ has finite first moment.
\end{proof}
\begin{lem}\label{lem:topologies-coincide}
  Let $(\mu_n)_{n\in\nat}$ be a sequence in $\Rad(X,d)$ and let $\mu\in\Rad(X,d)$.
  Then $\mu_n\to\mu$ in the topology given by $\dk(d)$ iff $L_{\mu_n}\to L_\mu$ in the weak-* topology.
\end{lem}
\begin{proof}
  By the definition of the Wasserstein distance $\dk(d)$
  we have $\mu_n\to\mu$ iff $\expect{\mu_n}{f} \to \expect{\mu}{f}$ for all nonexpansive maps $f\colon X\to\real$.
  As integration is linear, this holds iff $\expect{\mu_n}{f} \to \expect{\mu}{f}$ for all $f\in\Lip_0(X)$.
  This precisely means that $L_{\mu_n}\to L_\mu$ in the weak-* topology.
\end{proof}
\begin{lem}\label{lem:wasserstein-norm}
  For any $\mu,\nu\in\Rad(X,d)$ we have $\dk(d)(\mu,\nu) = \opnorm{\nu-\mu}$.
\end{lem}
\begin{proof}
  This follows more or less by expanding definitions:
  \begin{align*}
    \dk(d)(\mu,\nu) &= \sup \{ \expect{\nu}{f} - \expect{\mu}{f} \mid f \text{ nonexpansive} \} \\
    &= \sup \{ \sint{X}{f}{(\nu-\mu)} \mid f \text{ nonexpansive} \} = \opnorm{\nu-\mu}.\qedhere
  \end{align*}
\end{proof}

\begin{lem}\label{lem:kantorovich-ball}
  Let $\epsilon > 0$.
  For $x,y\in X$, $x\neq y$, let
  \begin{math}
    \textstyle\mu_{xy} = \frac{\epsilon}{d(x,y)} (\delta_y - \delta_x),
  \end{math}
  where~$\delta_x$ and~$\delta_y$ are Dirac measures, and put $C = \convx \{ \mu_{xy} \mid x,y\in X, x \neq y\}$.
  Then we have, for any $\mu,\nu\in\dfunpmet (X,d)$, that
  $\dk(d)(\mu,\nu) \le \epsilon \iff \nu-\mu \in \overline{C}$.
\end{lem}
\begin{proof}
  For `$\Leftarrow$', because the $\epsilon$-ball around $\mu$ is closed and convex, it suffices to consider the generators $\mu_{xy}$ of the convex hull.
  By Lemma~\ref{lem:wasserstein-norm}, we just need to show that $\opnorm{\mu_{xy}} \le \epsilon$:
  \begin{equation*}
    \opnorm{\mu_{xy}} = \sup_{\lipnorm{f} \le 1} \expect{\mu}{f}
    = \sup_{\lipnorm{f} \le 1} \frac{\epsilon}{d(x,y)} (f(y)-f(x))
    \le \sup_{\lipnorm{f} \le 1} \epsilon \lipnorm{f}
    \le \epsilon.
  \end{equation*}
  For `$\Rightarrow$', let $\dk(d)(\mu,\nu) \le \epsilon$.
  Because the finitely supported measures are a dense subset of the Radon measures~\cite{fritz_probability_2019}, we may assume without loss of generality that $\mu$ and $\nu$ have finite support.
  Let $\rho$ be an optimal coupling of $\mu$ and $\nu$, that is $\expect{\rho}{d} \le \epsilon$ and $\dfun\pi_1(\rho) = \mu$ and $\dfun\pi_2(\rho) = \nu$.
  We have $\nu-\mu = \sum_{x\neq y} \rho(x,y)(\delta_y - \delta_x) = \sum_{x\neq y} t_{xy} \mu_{xy}$, where $t_{xy} = \rho(x,y)\frac{d(x,y)}{\epsilon}$.
  We know that there are only finitely many $t_{xy}\neq 0$ by assumption, so it remains to show that $t_{xy}\ge 0$ and $\sum t_{xy} \le 1$, i.e.~we have a subconvex combination of the $\mu_{xy}$ (note that a subconvex combination suffices because we have $0\in C$).
  The former is clear from the definition, and for the latter,
  \begin{equation*}
    \textstyle \sum_{x\neq y} t_{xy} = \frac{1}{\epsilon}\sum_{x\neq y} \rho(x,y) d(x,y) = \frac{1}{\epsilon} \expect{\rho}{d} \le 1. \qedhere
  \end{equation*}
\end{proof}

\begin{thm}\label{thm:convex-powerset-duality-borel}
  Let $(X,d)$ be a pseudometric space and $A,B\in\cfunpmet(X,d)$.
  Then 
  \begin{equation*}
     \dhk(d)(A,B) = \sup \{ |\sup_{\nu\in B} \expect{\nu}{f} - \sup_{\mu\in A} \expect{\mu}{f}\,| 
    \mid f\colon (X,d) \to ([0,1],\de) \text{ nonexpansive} \}.
  \end{equation*}
\end{thm}
\begin{proof}
  We begin with the inequality `$\ge$'.
  It holds because $\dhk$ decomposes into two Kantorovich functors with respect to $\sup$ and $\expectSymb$ respectively, and because for every nonexpansive map $f\colon(X,d)\to([0,1],\de)$ the map $\cfunpmet(X,d)\to([0,1],\de)$, $\mu\mapsto\expect{\mu}{f}$ is nonexpansive as well, which implies that every term that is part of the supremum on the right is also part of the supremum on the left.

  For the other inequality `$\le$', we can reuse the proof of \cref{thm:convex-powerset-duality}, as there was no part of it that was specific to discrete probability measures;
  we just use \cref{lem:embed-into-dual,lem:topologies-coincide,lem:wasserstein-norm,lem:kantorovich-ball} instead of their specialized forms \cref{lem:embed-into-dual-discrete,lem:topologies-coincide-discrete,lem:wasserstein-norm-discrete,lem:kantorovich-ball-discrete}.
\end{proof}

\noindent
\Cref{thm:convex-powerset-duality-borel} establishes that the composition of the Hausdorff and Kantorovich-Rubinstein distances can be expressed in price-function form with respect to the single modality $A \mapsto \sup_{\mu\in A} \expect{\mu}{f}$.
To extend this into a full duality result for continuous measures, stating that $\dhk$ also coincides with the corresponding coupling-based construction, two more ingredients are needed:
First, we have duality for the two individual distances $\dh$ and $\dk$, as discussed in ~\cref{expl:kant-wass-instances}.
This is where the restriction from arbitrary Borel measures to Radon measures is necessary, as this is a prerequisite for the probabilistic duality result~\cite{Kellerer85}.
In the case of the Hausdorff distance, no additional requirements are necessary.
Second, we have compositionality of the respective coupling-based constructions:

\begin{lem}\label{lem:convex-wass-compose}
  We have $\wasspmet{\lambda}(d) = \wasspmet{\sup}(\wasspmet{\expectSymb}(d))$ for every pseudometric $d$.
\end{lem}
\begin{proof}
  Because for all of the involved functors the maps $F\langle\pi_1,\pi_2\rangle$ are in fact identity maps, once the occurrences of the predicate liftings $\lambda$ and $\sup$ are expanded, the terms are indistinguishable from those in the discrete case, except that we take expected values of Radon measures instead of discrete measures.
  No part of the proof of the compositionality result (\cite[Example 6.8.3]{FuzzyLaxHemi}) makes use of properties specific to discrete measures, so it can be reused without change.
\end{proof}

\noindent
Putting all of these ingredients together, we obtain duality:
\begin{thm}
  We have $\dhk(d) = \kantsym{\lambda}(d) = \wasspmet{\lambda}(d)$ for every pseudometric $d$.
\end{thm}
\begin{proof}
  The first equality is~\cref{thm:convex-powerset-duality-borel}.
  For the second equality, we have
  \begin{equation*}
    \wasspmet{\lambda}(d) = \wasspmet{\sup}(\wasspmet{\expectSymb}(d)) = \dhk(d),
  \end{equation*}
  where the first step is \cref{lem:convex-wass-compose} and the second step
  is by the duality results for the individual distances $\dh$ and $\dk$ (Example~\ref{expl:kant-wass-instances}).
\end{proof}

% \bknote[inline]{Correspondence Wasserstein/Kantorovich
%   \cite{FuzzyLaxHemi}}

\section{Conclusions and Future Work}
\label{sec:conclusion}

\noindent We have proved \emph{generalized Kantorovich-Rubinstein
duality}, i.e.~coincidence of coupling-based (or Wasserstein)
and price-function-based (or Kantorovich or codensity)
presentations of functor liftings induced by a given choice of
modalities, for two important and non-trivial cases: the
L\'evy-Prokhorov distance on distributions, and the standard distance
on convex sets of distributions that arises from composing the
Hausdorff and Kantorovich-Rubinstein metrics. In both cases, we
obtain a characterization of the respective distance by means of
quantitative modal logics defined by the given modalities; for the
case of the L\'evy-Prokhorov distances, this logic is (up to
restriction of the propositional base) the logic of \emph{generally}
previously studied in context of fuzzy description
logics~\cite{SchroderPattinson11}, and in the second case the involved
modality is just the composite of the usual fuzzy diamond and the
expectation modality~\cite{FuzzyLaxHemi}. In the case of convex
powerset, we demonstrate additionally that the price-function-based
presentation plays out algorithmic advantages in the actual
computation of distances.

One open problem that remains is to extend the duality result for
convex powerset to unrestricted fuzzy relations in place of
pseudometrics (already established in our result on the
L\'evy-Prokhorov metric). Such a generalization will amount to a
duality result for the known coupling-based lax extension of convex
powerset~\cite{FuzzyLaxHemi}. Also, we aim to capitalize on the
present results on convex powerset in the design of algorithms that
actually compute distinguishing formulae as witnesses of lower bounds
on behavioural distance, complementing recent results on behavioural
distance under the
Kantorovich-Rubinstein~\cite{rb:explainability-labelled-mc,TurkenburgEA26}
and L\'evy-Prokhorov~\cite{ForsterEA26} distance of distributions.

\bibliographystyle{splncs04}
\bibliography{references}

\providecommand{\noopsort}[1]{}
\begin{thebibliography}{10}
\providecommand{\url}[1]{\texttt{#1}}
\providecommand{\urlprefix}{URL }
\providecommand{\doi}[1]{https://doi.org/#1}

\bibitem{AdamekEA90}
Ad{\'a}mek, J., Herrlich, H., Strecker, G.E.: Abstract and concrete categories:
  {T}he joy of cats. John Wiley \& Sons Inc. (1990),
  \url{http://tac.mta.ca/tac/reprints/articles/17/tr17abs.html}, republished
  in: Reprints in Theory and Applications of Categories, No. 17 (2006)
  pp.~1--507

\bibitem{afs:linear-branching-metrics}
de~Alfaro, L., Faella, M., Stoelinga, M.: Linear and branching system metrics.
  {IEEE} Trans.\ Software Eng.  \textbf{35}(2),  258--273 (2009).
  \doi{10.1109/TSE.2008.106}

\bibitem{AolariteiEA25}
Aolaritei, L., Wang, O., Zhu, J., Jordan, M., Marzouk, Y.: Conformal prediction
  under {L}évy-{P}rokhorov distribution shifts: Robustness to local and global
  perturbations. In: Neural Information Processing Systems, NeurIPS 2025
  (2025), to appear. Preprint available at
  \texttt{https://arxiv.org/abs/2502.14105}

\bibitem{bbkk:coalgebraic-behavioral-metrics}
Baldan, P., Bonchi, F., Kerstan, H., K{\"{o}}nig, B.: Coalgebraic behavioral
  metrics. Log. Methods Comput. Sci.  \textbf{14}(3) (2018).
  \doi{10.23638/LMCS-14(3:20)2018}

\bibitem{BennounaEA23}
Bennouna, M.A., Lucas, R., Parys, B.P.G.V.: Certified robust neural networks:
  Generalization and corruption resistance. In: Krause, A., Brunskill, E., Cho,
  K., Engelhardt, B., Sabato, S., Scarlett, J. (eds.) International Conference
  on Machine Learning, {ICML} 2023. Proc.\ Machine Learning Res., vol.~202, pp.
  2092--2112. {PMLR} (2023),
  \url{https://proceedings.mlr.press/v202/bennouna23a.html}

\bibitem{bkp:up-to-behavioural-metrics-fibrations-journal}
Bonchi, F., König, B., Petrişan, D.: Up-to techniques for behavioural metrics
  via fibrations. Mathematical Structures in Computer Science
  \textbf{33}(4–5),  182–221 (2023). \doi{10.1017/S0960129523000166}

\bibitem{Bonchi_et_al:PowerConvexAlg}
Bonchi, F., Silva, A., Sokolova, A.: {The Power of Convex Algebras}. In: Meyer,
  R., Nestmann, U. (eds.) Concurrency Theory, CONCUR 2017. Leibniz
  International Proceedings in Informatics (LIPIcs), vol.~85, pp. 23:1--23:18.
  Schloss Dagstuhl--Leibniz-Zentrum fuer Informatik, Dagstuhl, Germany (2017).
  \doi{10.4230/LIPIcs.CONCUR.2017.23}

\bibitem{CelikEA20}
{\c{C}}elik, T.{\"{O}}., Jamneshan, A., Mont{\'{u}}far, G., Sturmfels, B.,
  Venturello, L.: Optimal transport to a variety. In: Slamanig, D., Tsigaridas,
  E.P., Zafeirakopoulos, Z. (eds.) Mathematical Aspects of Computer and
  Information Sciences, {MACIS} 2019. LNCS, vol. 11989, pp. 364--381. Springer
  (2019). \doi{10.1007/978-3-030-43120-4\_29}

\bibitem{dgkknrw:behavioural-metrics-compositionality}
D'Angelo, K., Gurke, S., Kirss, J.M., König, B., Najafi, M., R{\'o}{\.z}owski,
  W., Wild, P.: Behavioural metrics: Compositionality of the {K}antorovich
  lifting and an application to up-to techniques. In: Concurrency Theory,
  CONCUR 2024. {LIPIcs}, vol.~311, pp. 20:1--20:19. Schloss Dagstuhl -- Leibniz
  Center for Informatics (2024),
  \url{https://doi.org/10.4230/LIPIcs.CONCUR.2024.20}

\bibitem{DesharnaisEA08}
Desharnais, J., Laviolette, F., Tracol, M.: Approximate analysis of
  probabilistic processes: Logic, simulation and games. In: Quantitative
  Evaluation of Systems, QEST 2008. pp. 264--273. {IEEE} Computer Society
  (2008). \doi{10.1109/QEST.2008.42}

\bibitem{DesharnaisSokolova26}
Desharnais, J., Sokolova, A.: {\(\epsilon\)}-distance via
  {L}{\'{e}}vy-{P}rokhorov lifting. In: Guerrini, S., K{\"{o}}nig, B. (eds.)
  34th {EACSL} Annual Conference on Computer Science Logic, {CSL} 2026, Paris,
  France, February 23-28, 2026. LIPIcs, vol.~363, pp. 26:1--26:24. Schloss
  Dagstuhl - Leibniz-Zentrum f{\"{u}}r Informatik (2026).
  \doi{10.4230/LIPICS.CSL.2026.26},
  \url{https://doi.org/10.4230/LIPIcs.CSL.2026.26}

\bibitem{DudleyRealAnalysis}
Dudley, R.M.: Real Analysis and Probability. Cambridge Studies in Advanced
  Mathematics, Cambridge University Press, 2 edn. (2002)

\bibitem{Forster_et_al:CSL.2023:Density}
Forster, J., Goncharov, S., Hofmann, D., Nora, P., Schr\"{o}der, L., Wild, P.:
  Quantitative {H}ennessy-{M}ilner theorems via notions of density. In: Klin,
  B., Pimentel, E. (eds.) Computer Science Logic, CSL 2023. LIPIcs, vol.~252,
  pp. 22:1--22:20. Schloss Dagstuhl -- Leibniz-Zentrum f{\"u}r Informatik
  (2023). \doi{10.4230/LIPIcs.CSL.2023.22}

\bibitem{ForsterEA26}
Forster, J., Schr{\"{o}}der, L., Wild, P., K{\"{o}}nig, B., Nora, P.:
  Threshold-based behavioural distances. In: Sokolova, A., Totzke, P. (eds.)
  Concurrency Theory, {CONCUR} 2026. LIPIcs, vol.~391, pp. 35:1--35:19. Schloss
  Dagstuhl -- Leibniz-Zentrum f{\"{u}}r Informatik (2026).
  \doi{10.4230/LIPICS.CONCUR.2026.35}

\bibitem{fritz_probability_2019}
Fritz, T., Perrone, P.: A probability monad as the colimit of spaces of finite
  samples. Theory Appl.\ Categ.  \textbf{34},  170–220 (2019),
  \url{http://www.tac.mta.ca/tac/volumes/34/7/34-07.pdf}

\bibitem{KantorovichFunctors}
Goncharov, S., Hofmann, D., Nora, P., Schr{\"{o}}der, L., Wild, P.:
  {K}antorovich functors and characteristic logics for behavioural distances.
  In: Kupferman, O., Sobocinski, P. (eds.) Foundations of Software Science and
  Computation Structures, FoSSaCS 2023. LNCS, vol. 13992, pp. 46--67. Springer
  (2023). \doi{10.1007/978-3-031-30829-1\_3}

\bibitem{GoubaultLarrecq07Previsions}
Goubault{-}Larrecq, J.: Continuous previsions. In: Duparc, J., Henzinger, T.A.
  (eds.) Computer Science Logic, 21st International Workshop, {CSL} 2007, 16th
  Annual Conference of the EACSL, Lausanne, Switzerland, September 11-15, 2007,
  Proceedings. Lecture Notes in Computer Science, vol.~4646, pp. 542--557.
  Springer (2007). \doi{10.1007/978-3-540-74915-8\_40}

\bibitem{GoubaultLarrecq08}
Goubault{-}Larrecq, J.: Simulation hemi-metrics between infinite-state
  stochastic games. In: Amadio, R.M. (ed.) Foundations of Software Science and
  Computational Structures, 11th International Conference, {FoSSaCS} 2008.
  Lecture Notes in Computer Science, vol.~4962, pp. 50--65. Springer (2008).
  \doi{10.1007/978-3-540-78499-9\_5}

\bibitem{GoubaultLarrecq17}
Goubault{-}Larrecq, J.: Isomorphism theorems between models of mixed choice.
  Math. Struct. Comput. Sci.  \textbf{27}(6),  1032--1067 (2017).
  \doi{10.1017/S0960129515000547}

\bibitem{GoubaultLarrecqKRQM-I}
Goubault{-}Larrecq, J.: {K}antorovich-{R}ubinstein quasi-metrics i: Spaces of
  measures and of continuous valuations. Topology and its Applications
  \textbf{295},  107673 (2021). \doi{10.1016/j.topol.2021.107673}

\bibitem{HoffmannKunze71}
Hoffmann, K., Kunze, R.A.: Linear algebra. Prentice-Hall Hoboken, NJ (1971)

\bibitem{Hofmann07}
Hofmann, D.: Topological theories and closed objects. Adv.\ Math.
  \textbf{215}(2),  789 -- 824 (2007).
  \doi{https://doi.org/10.1016/j.aim.2007.04.013}

\bibitem{HofmannEA14}
Hofmann, D., Seal, G., Tholen, W.: Monoidal Topology: A Categorical Approach to
  Order, Metric, and Topology, vol.~153. Cambridge University Press (2014).
  \doi{10.1017/CBO9781107517288},
  \url{http://dx.doi.org/10.1017/CBO9781107517288}

\bibitem{HumeauEA25}
Humeau, S., Petrisan, D., Rot, J.: Correspondences between codensity and
  coupling-based liftings, a practical approach. In: Endrullis, J., Schmitz, S.
  (eds.) Computer Science Logic, {CSL} 2025. LIPIcs, vol.~326, pp. 29:1--29:18.
  Schloss Dagstuhl -- Leibniz-Zentrum f{\"u}r Informatik (2025).
  \doi{10.4230/LIPICS.CSL.2025.29}

\bibitem{Kantorovich39}
Kantorovich, L.V.: The mathematical method of production planning and
  organization. Management Science  \textbf{6}(4),  363--422 (1939)

\bibitem{Kellerer85}
Kellerer, H.G.: Duality theorems and probability metrics, pp. 211--220. De
  Gruyter (1985). \doi{doi:10.1515/9783112314036-025}

\bibitem{KomoridaEA19}
Komorida, Y., Katsumata, S., Hu, N., Klin, B., Hasuo, I.: Codensity games for
  bisimilarity. In: Logic in Computer Science, {LICS} 2019. pp. 1--13. {IEEE}
  (2019). \doi{10.1109/LICS.2019.8785691}

\bibitem{kkkrh:expressivity-quantitative-modal-logics}
Komorida, Y., Katsumata, S., Kupke, C., Rot, J., Hasuo, I.: Expressivity of
  quantitative modal logics: Categorical foundations via codensity and
  approximation. In: Logic in Computer Science, LICS 2021. pp. 1--14. {IEEE}
  (2021). \doi{10.1109/LICS52264.2021.9470656}

\bibitem{km:bisim-games-logics-metric}
K{\"o}nig, B., Mika-Michalski, C.: ({Metric}) bisimulation games and
  real-valued modal logics for coalgebras. In: Concurrency Theory, CONCUR 2018.
  LIPIcs, vol.~118, pp. 37:1--37:17. Schloss Dagstuhl -- Leibniz Center for
  Informatics (2018). \doi{10.4230/LIPICS.CONCUR.2018.37}

\bibitem{mv:monads-quantitative-equational}
Mio, M., Vignudelli, V.: Monads and quantitative equational theories for
  nondeterminism and probability. In: Konnov, I., Kov{\'{a}}cs, L. (eds.)
  Concurrency Theory, {CONCUR} 2020. LIPIcs, vol.~171, pp. 28:1--28:18. Schloss
  Dagstuhl -- Leibniz-Zentrum f{\"{u}}r Informatik (2020).
  \doi{10.4230/LIPICS.CONCUR.2020.28}

\bibitem{Pattinson04}
Pattinson, D.: Expressive logics for coalgebras via terminal sequence
  induction. Notre Dame J.\ Formal Log.  \textbf{45}(1),  19--33 (2004).
  \doi{10.1305/NDJFL/1094155277}

\bibitem{Prokhorov56}
Prokhorov, Y.V.: Convergence of random processes and limit theorems in
  probability theory. Theory of Probability \& Its Applications  \textbf{1}(2),
   157--214 (1956). \doi{10.1137/1101016}

\bibitem{rb:explainability-labelled-mc}
Rady, A., van Breugel, F.: Explainability of probabilistic bisimilarity
  distances for labelled {M}arkov chains. In: Kupferman, O., Sobocinski, P.
  (eds.) Foundations of Software Science and Computation Structures, FoSSaCS
  2023. LNCS, vol. 13992, pp. 285--307. Springer (2023).
  \doi{10.1007/978-3-031-30829-1\_14}

\bibitem{Rutten00}
Rutten, J.J.M.M.: Universal coalgebra: a theory of systems. Theor.\ Comput.\
  Sci.  \textbf{249}(1),  3--80 (2000). \doi{10.1016/S0304-3975(00)00056-6}

\bibitem{Schroeder2007:ExpressivityCoalgModalLog}
Schr\"{o}der, L.: Expressivity of coalgebraic modal logic: The limits and
  beyond. Theor.\ Comput.\ Sci.  \textbf{390}(2),  230--247 (2008).
  \doi{10.1016/j.tcs.2007.09.023}

\bibitem{SchroderPattinson11}
Schr{\"{o}}der, L., Pattinson, D.: Description logics and fuzzy probability.
  In: Walsh, T. (ed.) International Joint Conference on Artificial
  Intelligence, {IJCAI} 2011. pp. 1075--1081. {IJCAI/AAAI} (2011).
  \doi{10.5591/978-1-57735-516-8/IJCAI11-184}

\bibitem{TurkenburgEA26}
Turkenburg, R., Beohar, H., van Breugel, F., Kupke, C., Rot, J.: Constructing
  witnesses for lower bounds on behavioural distances. In: Computer Science
  Logic, {CSL} 2026 (2026), to appear. Preprint available at
  \texttt{https://arxiv.org/abs/2504.08639}

\bibitem{bw:behavioural-pseudometric}
{van Breugel}, F., Worrell, J.: A behavioural pseudometric for probabilistic
  transition systems. Theor. Comput. Sci.  \textbf{331}(1),  115--142 (2005).
  \doi{10.1016/J.TCS.2004.09.035}

\bibitem{v:optimal-transport}
Villani, C.: Optimal Transport -- Old and New. Springer (2009).
  \doi{10.1007/978-3-540-71050-9}

\bibitem{PaulWildThesis}
Wild, P.: The Model Theory of Quantitative Coalgebraic Modal Logics. Ph.D.
  thesis, Friedrich-Alexander-Universität Erlangen-Nürnberg (2024).
  \doi{10.25593/open-fau-480}

\bibitem{FuzzyLaxHemi}
Wild, P., Schr{\"{o}}der, L.: Characteristic logics for behavioural hemimetrics
  via fuzzy lax extensions. Log.\ Methods Comput.\ Sci.  \textbf{18}(2) (2022).
  \doi{10.46298/lmcs-18(2:19)2022}

\bibitem{WildEA26}
Wild, P., Schr{\"{o}}der, L., Messing, K., K{\"{o}}nig, B., Forster, J.:
  Generalized {K}antorovich-{R}ubinstein duality beyond {H}ausdorff and
  {K}antorovich. In: Bertrand, N., Milius, S. (eds.) Foundations of Software
  Science and Computation Structures, FoSSaCS 2026. LNCS, vol. 16503, pp.
  636--657. Springer (2026). \doi{10.1007/978-3-032-22730-0\_30}

\end{thebibliography}

\appendix

\end{document}